\documentclass[journal]{IEEEtran}

\newcommand{\avec}{{\bf{a}}}

\newcommand{\cvec}{{\bf{c}}}

\newcommand{\evec}{{\bf{e}}}

\newcommand{\yvec}{{\bf{y}}}

\newcommand{\xvec}{{\bf{x}}}
\newcommand{\zvec}{{\bf{z}}}

\newcommand{\rvec}{{\bf{r}}}

\newcommand{\zerovec}{{\bf{0}}}

\newcommand{\Amat}{{\bf{A}}}

\newcommand{\Fmat}{{\bf{F}}}

\newcommand{\Hmat}{{\bf{H}}}

\newcommand{\Imat}{{\bf{I}}}

\newcommand{\Kmat}{{\bf{K}}}
\newcommand{\Pmat}{{\bf{P}}}

\newcommand{\Smat}{{\bf{S}}}

\newcommand{\Rmat}{{\bf{R}}}

\newcommand{\define}{\stackrel{\triangle}{=}}

\def\thetavec{{\mbox{\boldmath $\theta$}}}

\newcommand{\be}{\begin{equation}}
\newcommand{\ee}{\end{equation}}
\newcommand{\beqna}{\begin{eqnarray}}
\newcommand{\eeqna}{\end{eqnarray}}

\newcommand{\edis}{{\evec\sim \mathcal{N}(\bf{0},\,\Rmat)}}

\newcommand{\Rupm}{{\mathbb{R}^{M}}}

\newcommand{\upperRomannumeral}[1]{\uppercase\expandafter{\romannumeral#1}}

\newcommand{\Hnull} {\mathcal{H}_{0}}
\newcommand{\Ha} {\mathcal{H}_{1}}

\newcommand{\bt}{\boldsymbol{\theta}}
\newcommand{\hbt}{\hat{\boldsymbol{\theta}}}

\usepackage{cite}
\usepackage{amsmath,amssymb,amsfonts,amsthm}
\usepackage{graphicx}
\usepackage{textcomp}
\usepackage{eucal}

\usepackage{algorithm,algorithmic} 
\usepackage{color,soul}
\usepackage{soul}
\usepackage{verbatim}
\usepackage{enumitem,lipsum}
\usepackage{mathptmx}
\usepackage{lscape,longtable,nomencl}
\usepackage{epsfig,subfigure}
\usepackage{mathtools}
\usepackage{caption}
\usepackage{stackrel}
\usepackage{setspace}
\usepackage[resetlabels,labeled]{multibib}
\usepackage{array}
\graphicspath{{figures/}} 
\usepackage{boldline,multirow}
\usepackage[utf8]{inputenc}
\usepackage[english]{babel}
\usepackage{lipsum}
\usepackage{bbm}

\newtheorem{thm}{Theorem}

\newtheorem{prop}[thm]{Proposition}

\theoremstyle{definition}

\newcommand{\RNum}[1]{\uppercase\expandafter{\romannumeral #1\relax}}

\newcommand{\bcolon}{\boldsymbol{:}}

\begin{document}

\captionsetup[figure]{labelfont={bf},labelformat={default},labelsep=period,name={Fig.}}

\title{Structural-constrained Methods for the Identification of False Data Injection Attacks in Power Systems}
\author{Gal Morgenstern,~\IEEEmembership{Student Member,~IEEE} and Tirza~Routtenberg, \IEEEmembership{Senior Member,~IEEE}
\vspace{-0.1cm}
\thanks{Gal Morgenstern and Tirza Routtenberg are with School of Electrical and Computer Engineering,   Ben-Gurion University of the Negev,
	Beer-Sheva 84105, Israel, Email: galmo@post.bgu.ac.il, tirzar@bgu.ac.il.}
\thanks{This research was partially supported by the 
	ISRAEL SCIENCE FOUNDATION (ISF), Grant No. 1173/16
	and by the BGU Cyber Security Research Center.
Gal Morgenstern has been funded by the Israeli Ministry of Science and Technology. }
}

\maketitle

\begin{abstract}
	Power system functionality is determined on the basis of power system state estimation (PSSE). 
	Thus, corruption of the PSSE may lead to severe consequences, such as disruptions in electricity distribution, maintenance damage, and financial losses.
	Classical bad data detection (BDD) methods, developed to ensure PSSE reliability, 
	are unable to detect well-designed attacks, named unobservable false data injection (FDI) attacks.
	In this paper, we develop novel structural-constrained methods for the detection of unobservable FDI attacks and the identification of the attacked buses.
	The proposed methods are based on formulating structural, sparse constraints on both the attack and the system loads.
	First, we exploit these constraints in order to compose an appropriate model selection problem.
	Then, we develop the associated generalized information criterion (GIC) for this problem.
	However, the GIC method's computational complexity
	grows exponentially with the network size, which may be prohibitive for large networks.
	Therefore, based on the proposed structural and sparse constraints, we develop two novel low-complexity methods for the practical identification of unobservable FDI attacks:
	1) a modification of the state-of-the-art orthogonal matching pursuit (OMP) method; 
	and 2) a method that utilizes the graph Markovian property in power systems, i.e. the second-neighbor relationship between the power data at the system's buses. 
	In order to analyze the performance of the proposed methods, the  appropriate  oracle and clairvoyant detectors are also derived.
	The performance of the proposed methods is evaluated on the IEEE-30 bus test case.
\end{abstract}


\section{Introduction}
\IEEEPARstart{M}{odern} 
electrical grids  are monitored by energy management systems (EMSs). 
The EMS evaluates the power system state estimation (PSSE) for multiple monitoring purposes,
including stability assessment, control, and security  \cite{gomez2004power,giannakis2013monitoring}. 
To ensure the reliability of the measurements, residual-based bad data detection (BDD) methods have been integrated into the EMS \cite{gomez2004power}.
However, BDD methods cannot detect well-designed attacks, 
named unobservable false data injection (FDI) attacks \cite{liu2011false,liang2017review}.
Unobservable FDI attacks are achieved by manipulating measurements based on the power network topology \cite{liu2011false},
where the topology matrix is either known or can be estimated from historical data \cite{kekatos2015online,grotas2019power,kim2014subspace,Halihal_Routtenberg_2022}. 
Unobservable FDI attacks may inflict severe damage by influencing the PSSE \cite{liang2017review,xie2011integrity,jia2014impact}.
Therefore, appropriate tools for the detection, identification,
and estimation of these attacks, that ensure the reliability of
the PSSE, are essential for obtaining high power quality and
maintaining stable power system operation.

The problem of detecting unobservable FDI attacks has been considered in the last decade.
Methods that strategically protect a basic set of measurements were proposed in \cite{bi2014graphical,deng2015defending,kim2011strategic}.
In the same context, using synchronized phasor measurement units (PMUs) has been suggested \cite{kim2011strategic,chen2006placement,kim2013phasor}.
However, these approaches require integrating additional hardware into the power system, which results in high cost, a long installation period, 
and additional security vulnerabilities \cite{mohammadi2018noncircular}. 
In addition, they may result in an unobservable system \cite{gomez2004power,dabush2021state}. 
Detection by the moving target defense technique, where the system configuration is actively changed,
has been proposed in \cite{rahman2014moving,liu2018reactance}.
However, in these methods, effective attack detection carries the cost of operating far from the optimal state
and, thereby, incurs economic losses \cite{liu2018reactance}.
Detection methods based on machine learning and data mining have been suggested in \cite{esmalifalak2017detecting,wang2017novel,he2017real,almutairy2021accurate},
but these methods require a large set of historic power system data, 
which is usually unavailable.
Detection based on load forecasting has been suggested in \cite{xu2017achieving}, 
yet obtaining a reliable forecast is not ensured and may require extensive resources \cite{gross1987short}.
In conclusion, all the above methods cannot serve as practical solutions for the detection of unobservable FDI attacks in power systems. 
Moreover, these methods do not provide attack identification, i.e. localization
of the attacked buses. 

In this paper, we suggest a novel compressive sensing (CS) parametric approach, based on structural constraints,
for the detection of unobservable FDI attacks and the identification of the attacked buses. 
Sparse recovery CS techniques have become a foremost research area in the last two decades 
(see, e.g. {\cite{Chen_Donoho_2001,Tropp_2004,elad2010sparse,donoho2005stable}} and references therein). 
In power systems, CS algorithms have been used for multi-line outage identification {\cite{Zhu_Giannakis_2012,Chen_Zhao_Goldsmith_Poor2014}},
gross error identification \cite{Kim_Raich_2016}, identification of imbalances \cite{routtenberg2017centralized}, and BDD \cite{Xu_Tang_2013}.
CS algorithms designed against unobservable FDI attack have been provided in \cite{liu2014detecting,hao2015sparse,gao2016identification}, 
where these methods exploit temporal correlations of the power measurements. 
However, the methods in \cite{liu2014detecting,hao2015sparse,gao2016identification} require multi-time measurements
and do not consider the difference between the structures of the unobservable FDI attack and of the system measurements. 

In this paper, 
we formulate structural and sparse constraints for both the FDI attack and the change in the power system loads
between two consecutive time samples.
Accordingly, 
we formulate a model selection problem, 
where in each model the attack is assumed to be represented by a different set of dictionary elements, i.e. the topology matrix columns. 
Based on the model selection formulation and the generalized information criterion (GIC) selection rule \cite{stoica2004model},
we develop a structural-constrained method 
for the detection, identification, and estimation of unobservable FDI attacks.
However, the proposed GIC-based method requires an exhaustive search, for which the required computational complexity increases 
exponentially with the network size, and is therefore practically infeasible.  
Thus, based on the proposed structural and sparse constraints, we develop two novel low-complexity methods for unobservable FDI attack identification:
1) a modification of the state-of-the-art orthogonal
matching pursuit (OMP) method \cite{mallat1993matching}; 
and 2) a method that is based on the graph Markovian property in power systems, i.e. second-neighbor relationship  \cite{sedghi2015statistical}.
Finally, we show by numerical simulation that the proposed GIC, OMP, and Graph Markovian GIC (GM-GIC) methods outperform the detection and identification performance of existing methods. 
In particular, the proposed {\em{parametric}} methods that are based on information criteria outperform our previous
{\em{non-parametric}} detection methods in \cite{drayer2019detection},\cite{drayer2018detection} that are heuristic in nature and are based on the Graph Fourier Transform (GFT). 

The main contribution of this paper is threefold. 
First, we present a new model that takes into account the physical, statistical, and structural properties of unobservable FDI attacks and the typical load changes in power systems. 
Second, by leveraging the proposed model, we derive a model-selection-based approach for unobservable FDI attack identification. 
Finally, we propose an OMP-based method and the novel low-complexity GM-GIC method that utilizes both the sparsity and the graphical properties of the problem 
in order to reduce complexity, while preserving high capabilities for the identification of unobservable FDI attacks. 
Further,  we demonstrate that the GM-GIC method can be used in the general context of sparse recovery of a graph signal 
from the outputs of a graph filter, 
and is not limited to power system applications.

The remainder of this paper is organized as follows. 
In Section \ref{sec; Background} we introduce the necessary background on the power flow equations, BDD, and unobservable FDI attacks. 
The proposed structural-constrained model for unobservable FDI attacks is presented in Section \ref{sec; structural model}. 
This model is then used in Section \ref{sec; detection and identification} to develop the  GIC-based approach for the detection and identification of unobservable FDI attacks, which is the basis for the two low-complexity methods developed in Section \ref{sec; large scale}.  
Next, a simulation study is presented in Section \ref{sec; Simulation study} and conclusions in Section \ref{sec; Conclusions}.

In this paper vectors are denoted by boldface lowercase letters
and matrices are denoted by boldface uppercase letters.
The operators $\lvert\lvert \cdot \rvert\rvert$ and $\lvert\lvert \cdot \rvert\rvert_0$ 
denote the Euclidean norm and the zero seminorm, respectively, 
where the latter specifies the number of non-zero elements in the vector. 
The operators  $(\cdot)^T$ and $(\cdot)^{-1}$ are the transpose and inverse operators, respectively. 
The linear space spanned by the $\Amat$ matrix columns is denoted by $\text{col}(\Amat)$.
For an index set, $\Lambda \subset \{1,\dots,L\}$, 
$\thetavec_{\Lambda}$ is the $\lvert\Lambda\rvert$-dimensional subvector of $\thetavec$ containing the elements indexed by $\Lambda$, 
where  $\lvert\Lambda\rvert$ denotes the set's cardinality. 
For any index sets, $\Lambda_1$ and $\Lambda_2$, $\Amat_{\Lambda_1,\Lambda_2}$ is the submatrix 
composed by the columns and rows of $\Amat$ associated with  $\Lambda_1$ and $\Lambda_2$, respectively. 

\section{Background}
\label{sec; Background}
\subsection{System model} 
\label{sec;Sub; DC model}
A power system can be represented as an undirected weighted graph, ${\mathcal{G}}({\mathcal{V}},\mathcal{E})$, where the set of vertices, $\mathcal{V}$, 
is the set of buses (substations), and the edge set, $\mathcal{E}$, is the set of transmission lines between these buses.
The weight in each line is defined according to the $\pi$-model of transmission lines \cite{gomez2004power}.
Hence, the weight over line $(n,k)\in\mathcal{E}$ is given by the admittance of the transmission line, $Y_{n,k} $.
Specifically, in the direct current (DC) model, where branches are without resistance loss, only the imaginary part of the admittance (the susceptance) is considered as the line weight. 

The direct current (DC) power flow model  is a linearized representation of the power measurements.
This model
defines the relation between the active power measurements in the buses and power flows in the transmission lines, $\zvec=[z_1,\cdots,z_{\scriptscriptstyle{M}}]^T\in\Rupm$,
and the voltage angles (``states") at the $N$ buses, $\bt=[\theta_1,\ldots,\theta_{N}]^T$, in the power system network \cite{gomez2004power}. 
Based on the DC model, the standard observation model for data injection attacks is given by \cite{cui2012coordinated}:
\be \label{eq; observation model}
\begin{gathered}
	\zvec=\Hmat\bt+\avec+\evec. 
\end{gathered}
\ee
The topology matrix, $\Hmat$, is a constant $M\times N$ Jacobian matrix, $N<M$,
which is composed of the susceptance elements
(as described, for example, in \cite{kosut2011malicious,gomez2004power}). 
In addition, $\evec\in\Rupm$ is a zero-mean Gaussian additive noise vector with covariance matrix $\Rmat$, $\edis$.  
The attack is denoted by $\avec\in\Rupm$.
It is assumed that $\Hmat$ is a full-rank matrix,
i.e. ${\text{rank}(\Hmat)=N-1}$, and that the system is fully observable.

\subsection{Power system state estimation (PSSE)}
\label{sec; sub; PSSE and Bad data detection}
The PSSE is commonly used for multiple monitoring purposes,
where under the DC model, the system states are the voltage angles, $\bt$. 
The PSSE is commonly computed using the weighted least squares (WLS) estimator \cite{gomez2004power}:
\be\label{eq; WLS}
\hbt^{\scriptstyle\text{WLS}} =  \underset{\bt\in\mathbb{R}^N}{\text{argmin}} ~ (\zvec-\Hmat\bt)^{T}\Rmat^{-1}(\zvec-\Hmat\bt)  
= \Kmat\zvec,
\ee 
where 
\be \label{eq; K matrix}
\Kmat=(\Hmat^{T}\Rmat^{-1}\Hmat)^{-1}\Hmat^{T}\Rmat^{-1}.
\ee

In order to robustify the PSSE against errors, classical BDD methods are implemented \cite{gomez2004power}. 
These BDD methods are based on the residual error:
\be \label{eq; residual error}
\rvec=\zvec-\Hmat	\hbt^{\scriptstyle\text{WLS}}= (\Imat-\Hmat\Kmat)\zvec,
\ee
where the last equality is obtained by substituting the estimated state vector as defined in (\ref{eq; WLS}).
The residual error from \eqref{eq; residual error} is used, for example, 
in the Largest Normalized Residual $\rvec_{\text{max}}^N$-test 
and the $\chi^2$-test\cite{gomez2004power}, 
which is given by
\be \label{eq; BDD}
T^{\scriptstyle{\text{BDD}}} = \lvert\lvert \zvec- \Hmat	\hbt^{\scriptstyle\text{WLS}} \rvert\rvert^2  \mathop{\gtrless}_{\Hnull}^{\mathcal{H}_1} \gamma^{\scriptstyle{\text{BDD}}} ,
\ee
where $\Ha$ is the hypothesis that the measurements are corrupted by bad data, which may be  the  result of an attack, and $\Hnull$ is the null  hypothesis.
The threshold, $\gamma^{\scriptstyle{\text{BDD}}}$, is determined to obtain a desired false alarm  probability. 
However, both the Largest Normalized Residual, which is an identification method, and the $\chi^2$-test, 
which is a detection method, can neither detect nor identify the presence of unobservable FDI attacks, as described in the next subsection.

\subsection{ Unobservable FDI attacks}
\label{sec; sub; unobservable FDI attacks}
The unobservable FDI attack satisfies
\be \label {eq; FDIA}
\avec\define \Hmat\cvec,
\ee
where $\cvec\in\mathbb{R}^N$ is an arbitrary constant vector.
By substituting the well-designed attack from \eqref{eq; FDIA} in the observation model from \eqref{eq; observation model} we obtain 
\be \label{eq; FDIA model}
\zvec=\Hmat(\bt+\cvec)+\evec.
\ee
From a comparison between the models in \eqref{eq; observation model} and in \eqref{eq; FDIA model}, 
it can be verified that by observing $\zvec$, the state vector in \eqref{eq; observation model}, $\bt$, cannot be distinguished from 
its corrupted (attacked) version in \eqref{eq; FDIA model}, $\bt+\cvec$, since both $\bt$ and $\cvec$ are unknown vectors.
As a result, the residual error, obtained by substituting the WLS estimation 
(which is based on the unobservable FDI model from \eqref{eq; FDIA model}) in \eqref{eq; residual error} is 
\be \label{eq; dddd}
\rvec= (\Imat-\Hmat\Kmat)\evec,
\ee 
and therefore cannot be utilized in order to indicate the presence of an unobservable FDI attack.
Consequently, all residual-based methods,
as well as all the other methods that are based solely on the model in \eqref{eq; FDIA model}, are expected to fail in detecting unobservable FDI attacks.
Therefore, methods for the identification of manipulated measurements
that integrate additional constraints are required.

\section{Structural-constrained modeling for unobservable FDI attack identification}
\label {sec; structural model} 
Due to  the unobservable FDI attack formation and further sparsity restrictions, 
the additional power caused by the attack is constrained to a small subset of buses in the network.
Hence, in this section, a new framework that is based on observing power changes
in small subnetworks of the system is proposed.
First, the assumptions behind this framework are presented in Subsection \ref{sec; assumptions}. 
Then,  the new structural-constrained  model for power measurements in the presence of unobservable FDI attacks is developed in Subsection \ref{sec; framework}.

\subsection{Assumptions }
\label{sec; assumptions}
The proposed framework,
which facilitates structural constraints on the unobservable FDI attack and the typical load changes,
is constructed by applying the following definitions and assumptions:
\begin{enumerate}[label=\textbf{A.\arabic*},leftmargin=0.9cm]
	\item \label{A; diff} 
	\textbf{Difference-based model:}
	Similarly to the models in \cite{xu2017achieving} and \cite{jiang2017defense}, we assume that two consecutive time samples are observed from the model in \eqref{eq; FDIA model} 
	at times $t$ and $t+1$.  
	The first, $\zvec_t$, is free from malicious attacks, i.e. $\cvec_{t}=\zerovec$, while the second, $\zvec_{t+1}$, may contain an unobservable FDI attack, 
	i.e. $\cvec_{t+1}=\cvec$. 
	Thus, we obtain 
	\be \label{eq; attack difference} 
	\Delta\cvec\define\cvec_{t+1}-\cvec_{t}=\cvec,
	\ee
	where $\cvec=0$ if there is no attack.
	It should be noted that this assumption is not restrictive since, in practice, the proposed detection method can be integrated into an adaptive scheme in similar manner to 
	the detection method outlined in \cite{jiang2017defense}.
	The adaptive scheme is initialized by a secure (free of an unobservable FDI attack) measurement, and then the system is constantly monitored by observing the difference between two consecutive measurements. 
	Consequently, only a single measurement used for initializing the adaptive scheme is required to be free of an attack.

	\item \label{A; restricted}
	\textbf{Restricted measurements:} 
	Since generator buses are heavily secured and obtain direct communication to the control center \cite{yuan2011modeling}, \cite{tan2017modeling},
	we assume that generator bus measurements cannot be manipulated.
	Additionally, it is assumed that `zero load' measurements cannot be manipulated.  
	Mathematically, this assumption implies that
	\be
	\Hmat_{\{\{\mathcal{V}\setminus\mathcal{L}\},\mathcal{V}\}}\cvec=\zerovec,
	\ee 
	where the set $\{\mathcal{V}\setminus\mathcal{L}\}$ includes the generator and `zero load' buses, 
	i.e. all buses except the load buses stored in $\mathcal{L}$. 
	
	\item \label{A; subsets}
	\textbf{State sparsity:}
	The number of manipulated state variables is  bounded by the sparsity term $K_{\cvec}$, 
	which is considered to  be significantly smaller than the cardinality of the node set, $\lvert\mathcal{V}\rvert$, i.e. 
	$K_{\cvec}\ll\lvert \mathcal{V}\rvert$. 
	Following this sparsity restriction,  we define the set 
	\be \label{eq; G K C}
	\mathcal{G}_{K_{\cvec}}\define\{\Lambda\subset \mathcal{V}:~1\le \lvert \Lambda \rvert  \le K_{\cvec}  \},
	\ee
	that includes all possible supports of the state attack $\cvec$.
	As a result, there exists a node subset ${\Lambda_i\in	\mathcal{G}_{K_{\cvec}}}$, where $i=1,2,\cdots,\lvert\mathcal{G}_{K_{\cvec}}\rvert$, that fully contains the attack and satisfies
	\be \label{eq; lambda equal}
	\Hmat_{\bcolon, \mathcal{V} }\cvec=	\Hmat_{\bcolon,\Lambda_i}\cvec_{\Lambda_i}.
	\ee
	It is shown in \cite{liu2011false}, that this assumption  (also used in \cite{gao2016identification})   
	stems directly from the   commonly-used sparsity restriction on the number of manipulated meters, 
	which states that the attack vector, $\avec$, is sparse (see, e.g. \cite{kosut2011malicious}).  
	
	\item \label{A; quasi}
	\textbf{Typical load changes (quasi-static system): } 
	Power systems under normal conditions are quasi-static systems that only change slightly over a short period of time \cite{monticelli2000electric, liu2014detecting}. 
	Therefore, it is considered that   typical load changes satisfy \color{black}
	\be  \label{eq; quasi-static}
	\lvert\lvert \Hmat_{\mathcal{L},\mathcal{V}} \Delta\bt \rvert\rvert^2<\eta,
	\ee
	where  $\eta$ is  a relatively small  tuning parameter  that can be obtained from the system statistics.
	\item \label{A; loads}
	\textbf{Typical load changes (structural properties): } 
	The typical load changes w.r.t. the actual load measurements at a specific moment
	are determined by the consumption demand \cite{gross1987short}.
	Thus, representing the typical load changes in a matrix form will output a non-sparse vector that is not related to the system topology matrix.
	As a result, composing the typical load changes requires all of the topology matrix columns.
	Moreover, a close representation can be achieved only by using the vast majority of the topology matrix columns, which is equivalent to excluding
	only a small subset of topology matrix columns.
	Mathematically, we assume that for any $\Lambda_i\in\mathcal{G}_{K_{\cvec}}$, 
	there exists a small non-negative parameter $\epsilon_i$ such that 
	the orthogonal projection matrix, $\Pmat_{\Lambda_i}^{\bot}$, is an $(1-\epsilon_i)$ $\ell_2$-subspace embedding of $\Hmat_{\mathcal{L},\mathcal{V}}$  (see Definition 1 in \cite{woodruff2014sketching}, Subsection 2.1):
	\be \label{eq; embedding}
	\lvert\lvert \Pmat_{\Lambda_i}^{\bot} \Hmat_{\mathcal{L},\mathcal{V}} \Delta\bt \rvert\rvert^2=(1- \epsilon_i) 	 \lvert\lvert  \Hmat_{\mathcal{L},\mathcal{V}} \Delta\bt \rvert\rvert^2,
	\ee
	where 
	\be \label{eq; projection matrix}
	\Pmat_{\Lambda_i}\define\Hmat_{\mathcal{L},\Lambda_i} ( \Hmat_{\mathcal{L},\Lambda_i} ^T \Hmat_{\mathcal{L},\Lambda_i} )^{-1}   \Hmat_{\mathcal{L},\Lambda_i}^T
	\ee
	is the projection matrix onto the space $\text{col}(\Hmat_{\mathcal{L},\Lambda_i})$
	and  $\Pmat_{\Lambda_i}^\bot\define\Imat-\Pmat_{\Lambda_i}$ is the projection matrix onto the orthogonal space $(\text{col}(\Hmat_{\mathcal{L},\Lambda_i}))^{\bot}$.
	
\end{enumerate}
It can be seen that Assumptions \ref{A; diff}-\ref{A; subsets} define constraints on the unobservable FDI attack while
Assumptions \ref{A; quasi}-\ref{A; loads}, 
define constraints on the typical load changes.

\subsection{Structural-constrained model}
\label{sec; framework}
In the following, we use Assumptions \ref{A; diff}-\ref{A; loads} to construct the new model. First,
following Assumption \ref{A; diff}, we consider the difference-based model by taking the difference between two consecutive time samples
that satisfy the model in \eqref{eq; FDIA model},
with $\cvec=\boldsymbol{0}$ for the first time sample and an arbitrary $\cvec$ in the second time sample.
Thus, the difference-based observation model is given by
\be \label{eq; 1}
\Delta\zvec=\Hmat(\cvec+\Delta\bt)+\Delta\evec,
\ee
where $\Delta\bt$ is the change in the state vector, $\bt$, between the two consecutive time samples.
Furthermore, by assuming that the noises of the two consecutive time samples are independent, 
we obtain that the difference-based measurement noise in \eqref{eq; 1} is distributed according to ${\Delta\evec\sim \mathcal{N}(\boldsymbol{0},2\Rmat)}$,
where, for the sake of simplicity, it is assumed that $2\Rmat=\sigma_e^2\Imat$.

According to Assumption  \ref{A; restricted},
only the changes in the load buses' measurements are relevant for the attack.
Thus, the relevant measurement model for the detection of attacks is
obtained by taking only the load measurements from \eqref{eq; 1},
which results in
\be \label{eq; 2}
\Delta\zvec_{\mathcal{L}}=\Hmat_{\mathcal{L},\mathcal{V}}(\cvec+\Delta\bt)+\Delta\evec_{\mathcal{L}},
\ee
where $\Hmat_{\mathcal{L},\mathcal{V}}$ is the associated submatrix of the topology matrix $\Hmat$, 
and $\Delta\evec_{\mathcal{L}}$ is the corresponding noise. 

Assumption \ref{A; subsets} implies that the state attack vector support, $\Lambda$, is in the set $\mathcal{G}_{K_{\cvec}}$.
Hence, as shown in \eqref{eq; lambda equal}, there exists a subset of nodes, ${\Lambda_i\in	\mathcal{G}_{K_{\cvec}}}$, that fully contains the attack.
By substituting  \eqref{eq; lambda equal}
in \eqref{eq; 2} we obtain 
\be \label{eq; 3}
\Delta\zvec_{\mathcal{L}}=\Hmat_{\mathcal{L},\Lambda_i}\cvec_{\Lambda_i} +\Hmat_{\mathcal{L},\mathcal{V}} \Delta\bt+\Delta\evec_{\mathcal{L}}.
\ee
Thus, the projection of the measurement vector in \eqref{eq; 3} onto $\text{col}(\Hmat_{\mathcal{L},\Lambda_i})$ satisfies
\be \label{eq; model before GLRT}
\begin{aligned}
	\Pmat_{\Lambda_i}\Delta\zvec_{\mathcal{L}}&= \Pmat_{\Lambda_i}(\Hmat_{\mathcal{L},\Lambda_i}\cvec_{\Lambda_i} +\Hmat_{\mathcal{L},\mathcal{V}} \Delta\bt+\Delta\evec_{\mathcal{L}}) \\
	&=    \Hmat_{\mathcal{L},\Lambda_i}\cvec_{\Lambda_i} +\Pmat_{\Lambda_i}(\Hmat_{\mathcal{L},\mathcal{V}} \Delta\bt+\Delta\evec_{\mathcal{L}}), 
\end{aligned}
\ee
where $\Pmat_{\Lambda_i}\Hmat_{\mathcal{L},\Lambda_i}\cvec_{\Lambda_i}= \Hmat_{\mathcal{L},\Lambda_i}\cvec_{\Lambda_i}$ results from \eqref{eq; projection matrix}. 
In addition, by substituting the following property of projection matrices (see, e.g. p. 46 in \cite{yanai2011projection}):
\be \label{eq; projection property}
\lvert\lvert \Pmat_{\Lambda_i}^{\bot} \Hmat_{\mathcal{L},\mathcal{V}} \Delta\bt \rvert\rvert^2= \lvert\lvert  \Hmat_{\mathcal{L},\mathcal{V}} \Delta\bt \rvert\rvert^2-  \lvert\lvert \Pmat_{\Lambda_i} \Hmat_{\mathcal{L},\mathcal{V}} \Delta\bt \rvert\rvert^2
\ee
in \eqref{eq; embedding} from Assumption \ref{A; loads} and rearranging the equation, we obtain
\be \label{eq; approximation error}
\epsilon_i= \frac{ \lvert\lvert \Pmat_{\Lambda_i} \Hmat_{\mathcal{L},\mathcal{V}} \Delta\bt \rvert\rvert^2}{ \lvert\lvert  \Hmat_{\mathcal{L},\mathcal{V}} \Delta\bt \rvert\rvert^2}.
\ee
Hence, by substituting \eqref{eq; approximation error} in \eqref{eq; quasi-static}, we obtain 
that the projection of $ \Hmat_{\mathcal{L},\mathcal{V}} \Delta\bt $ onto the column space of the submatrix, $\Hmat_{\mathcal{L},\Lambda_i}$,
is bounded by 
\be  \label{eq; projection loads}
\lvert\lvert  \Pmat_{\Lambda_i}  \Hmat_{\mathcal{L},\mathcal{V}} \Delta\bt \rvert\rvert^2  \le \epsilon_i\eta,
\ee
where, according to Assumptions \ref{A; quasi}-\ref{A; loads}, $\epsilon_i\eta$ is a small parameter for any $i=1,\ldots,\lvert\mathcal{G}_{K_{\cvec}}\rvert$.
\color{black}

Based on \eqref{eq; 3} and \eqref{eq; projection loads}, identifying
the subset of attacked buses, $\Lambda_i$, under the considered model can be formulated as the following multiple hypothesis testing problem: 
\be \label{eq; model selection problem the}
\begin{aligned}
	&\mathcal{H}_0:~~~  \Delta\zvec_{\mathcal{L}}=\Hmat_{\mathcal{L},\mathcal{V}} \Delta\bt+\Delta\evec_{\mathcal{L}} \\[1.5pt]
	&\mathcal{H}_i:~~~ \begin{cases}
		\Delta\zvec_{\mathcal{L}}=\Hmat_{\mathcal{L},\Lambda_i}\cvec_{\Lambda_i} +\Hmat_{\mathcal{L},\mathcal{V}} \Delta\bt+\Delta\evec_{\mathcal{L}} \\
		\text{s.t.}~   \lvert \lvert\Pmat_{\Lambda_i}  \Hmat_{\mathcal{L},\mathcal{V}} \Delta\bt \rvert\rvert^2 \le \epsilon_i\eta, 
	\end{cases}
	&                                          
\end{aligned}
\ee
where $\epsilon_i\eta$ is small for any  $i=1,\ldots,\lvert\mathcal{G}_{K_{\cvec}}\rvert$.
Each hypothesis $\mathcal{H}_i$ in \eqref{eq; model selection problem the}
assumes a different support, $\Lambda_i\in\mathcal{G}_{K_{\cvec}}$, for the state attack vector, $\cvec$.
The null hypothesis $\mathcal{H}_0$ is obtained by substituting $\cvec=\boldsymbol{0}$ in \eqref{eq; 3}.   

Since $\epsilon\eta_i$ is small, 
we adopt standard practices from the CS
literature \cite{Chen_Donoho_2001,Tropp_2004,elad2010sparse, donoho2005stable}, where it is common to exclude low amplitude samples from the sparse approximation in order to develop tractable methods. 
That is, instead of solving \eqref{eq; model selection problem the} one can solve  the following modified hypothesis testing problem:
\be \label{eq; model selection problem}
\begin{aligned}
	&\mathcal{H}_0:~~~  \Delta\zvec_{\mathcal{L}}=\Hmat_{\mathcal{L},\mathcal{V}} \Delta\bt+\Delta\evec_{\mathcal{L}} \\[1.5pt]
	&\mathcal{H}_i:~~~ \begin{cases}
		\Delta\zvec_{\mathcal{L}}=\Hmat_{\mathcal{L},\Lambda_i}\cvec_{\Lambda_i} +\Hmat_{\mathcal{L},\mathcal{V}} \Delta\bt+\Delta\evec_{\mathcal{L}} \\
		\text{s.t.}~   \Pmat_{\Lambda_i}  \Hmat_{\mathcal{L},\mathcal{V}} \Delta\bt= \zerovec,
	\end{cases}
	&                                          
\end{aligned}
\ee
where $i=1,\ldots,\lvert\mathcal{G}_{K_{\cvec}}\rvert$. 
The multiple hypothesis testing in \eqref{eq; model selection problem} provides a new framework for identifying the sparse  state attack vector $\cvec$
from measurements contaminated by additive noise and the nuisance parameter $\Delta\bt$. 
In the modified hypothesis testing   in \eqref{eq; model selection problem},  
the state attack vector $\cvec$ and the load changes $\Delta\bt$ are not in the same subspace. 
Hence, in contrast to the use of \eqref{eq; FDIA model}, under the framework in \eqref{eq; model selection problem} the attack is observable. 
Consequently, the formulation in \eqref{eq; model selection problem} is appropriate for the development of unobservable FDI attack identification algorithms such as
the GIC-based identification algorithm developed in Section \ref{sec; detection and identification} and
the  low-complexity practical algorithms described in Section \ref{sec; large scale}. 
The performance of the proposed methods is examined w.r.t. the values $\eta$ and $\epsilon_i,~i=1,\ldots \lvert\mathcal{G}_{K_{\cvec}}\rvert$, from the hypothesis testing in \eqref{eq; model selection problem the}
in Subsection \ref{sec; performance} and empirically in Section \ref{sec; Simulation study} (see Fig. \ref{fig; sigma}). 
\color{black}

A possible interpretation of the  problem in \eqref{eq; model selection problem the} and \eqref{eq; model selection problem} is as a special case of Matched Subspace Detection (MSD) \cite{scharf1994matched}, 
in which the measurements are composed of two deterministic signals and additive noise. 
In a similar manner to in our framework, one of the two deterministic signals is considered as the target signal (here, the attack) and the other is considered as the background signal (here, the  load changes).
However, our framework deviates from the standard MSD problem by: 1) the spanning subspace of the target signal is unknown; 
2) sparsity restrictions are assumed on the target signal; and 3) the subspace of the target is contained within the linear space that spans the background signal. 
Thus, standard MSD techniques cannot be applied to solve  \eqref{eq; model selection problem the} or \eqref{eq; model selection problem}. 
A different perspective, presented in
Subsection \ref{sec; GSP extensions}, is the applicativity of \eqref{eq; model selection problem} in the context of graph signal processing (GSP), where $\Hmat$ is a general graph filter.

\section{Identification of unobservable FDI attacks by the GIC approach}
\label{sec; detection and identification}
\subsection{GIC approach}
In the following, we derive the identification of unobservable FDI attacks by selecting the most likely hypothesis in \eqref{eq; model selection problem}, 
i.e. the most suitable choice of attack support, $\Lambda_i$, from the set of candidate
supports,  $\{\mathcal{G}_{K_{\cvec}}\cup\emptyset\}$,  given the measurement vector, $\Delta\zvec_{\mathcal{L}}$.
In order to solve \eqref{eq; model selection problem}, we implement the GIC selection rule \cite{stoica2004model},
which is widely employed in signal and array processing.
The GIC method chooses
the hypothesis $\mathcal{H}_i$ which maximizes the sum of two terms: the
likelihood term for data encoding,
$L(\cdot)$,
and a penalty function, $\tau(\cdot)$, that inhibits the number of free parameters of the model from becoming very large. 
For the considered hypothesis testing in \eqref{eq; model selection problem} and a given difference-based state vector $\Delta\bt$, the GIC statistic is given by
\be \label{eq; GIC function}
\text{GIC}(\Lambda_i,\tau(\lvert\Lambda_i\rvert, \lvert\mathcal{L}\lvert) ) \define  2L(\Delta\zvec_{\mathcal{L}}; \hat{\cvec}^{\text{ML}|i}_{\Lambda_i},\Delta\bt)-\tau(\lvert\Lambda_i\rvert, \lvert\mathcal{L}\lvert) ,
\ee
where $L(\Delta\zvec_{\mathcal{L}}; \hat{\cvec}^{\text{ML}|i}_{\Lambda_i},\Delta\bt)$ is the log-likelihood function of $\Delta\zvec$
under the $i$th hypothesis, which is associated with the support, $\Lambda_i$,
and $ \hat{\cvec}^{\text{ML}|i}$ is the  ML estimation of the state attack vector, $\cvec$.
Therefore, in the general case the GIC statistic is a function of $\Delta\bt$, which is an unknown deterministic vector.
However, as presented in the following,
the GIC selection rule in this case is independent of  $\Delta\bt$ and, therefore, is a valid rule. 
The term $\tau(\lvert\Lambda_i\rvert, \lvert\mathcal{L}\lvert) $ is a penalty function,
which increases with the number of free unknown parameters 
that is determined in this case by the number of manipulated variables, $\lvert \Lambda_i\rvert$,
and the number of load buses, $\lvert\mathcal{L}\rvert$.
In particular, a special case of the GIC family is the Akaike information criterion (AIC), 
for which the penalty term is
\be \label{eq; MDL AIC}
\tau(\lvert\Lambda_i\rvert, \lvert\mathcal{L}\lvert) =2\lvert\Lambda_i\rvert.
\ee
Further discussion on the penalty term is provided in the simulation study in Section \ref{sec; Simulation study}. 

\begin{prop} \label{thm; GIC independent}
	The GIC statistic in \eqref{eq; GIC function} satisfies
	\be \label{eq; GIC}
	\begin{aligned}
		\text{\normalfont{GIC}}(\Lambda_i,\tau(\lvert\Lambda_i\rvert, \lvert\mathcal{L}\lvert) ) &= \frac{1}{\sigma_{\evec}^2}\lvert\lvert \Pmat_{\Lambda_i}\Delta\zvec_{\mathcal{L}}  \rvert\rvert^2 -\tau(\lvert\Lambda_i\rvert, \lvert\mathcal{L}\lvert) \\
		&~~~~+const,
	\end{aligned}
	\ee
	where $const$ is a constant term that does not depend on the hypothesis, ${i=0,1,\ldots, \lvert \mathcal{G}_{K_{\cvec}} \rvert}$.
\end{prop}
\begin{proof}
	The proof is provided in Appendix \ref{App; proof log likelihood}.
\end{proof}

It should be noted that for the null hypothesis in Proposition \ref{thm; GIC independent}, in which $\Lambda_0=\emptyset$, we use the convention that $ \Pmat_{\Lambda_0}\Delta\zvec_{\mathcal{L}} =\boldsymbol{0}$.
It can be seen from the GIC statistic in \eqref{eq; GIC} that the selected hypothesis is the one 
that maximizes the sum of the two terms: 1) the projection of the load measurements onto the associated ``attack subspace'', $\text{col}(\Hmat_{\mathcal{L},\Lambda_i})$,
by computing $ \Pmat_{\Lambda_i}\Delta\zvec_{\mathcal{L}}$; and 2) a penalty function, $ -\tau(\lvert\Lambda_i\rvert, \lvert\mathcal{L}\lvert)$, that encourages a sparse solution.
As a result, the representation of the GIC selection rule from \eqref{eq; GIC function} by \eqref{eq; GIC} 
shows that it is not a function of the unknown states, $\Delta\bt$, and therefore, it is a valid selection rule. 
An intuition regarding the GIC selection rule in Proposition \ref{thm; GIC independent} can be drawn from a comparison of the GIC statistic in \eqref{eq; GIC} with the bound in \eqref{eq; projection loads}. This comparison 
implies that  a high energy level for the projected signal in \eqref{eq; model before GLRT}, $\lvert\lvert \Pmat_{\Lambda_i}\Delta\zvec_{\mathcal{L}}\rvert\rvert^2$, cannot be associated with conventional states, and thus
provides an indication of the presence of an unobservable FDI attack in the subspace associated with $\Lambda_i$, $\Hmat_{\mathcal{L},\Lambda_i}\cvec_{\Lambda_i}$.

The proposed structural-constrained GIC algorithm for the identification of unobservable FDI attacks is provided in Algorithm \ref{alg; identification}
and is based on the GIC statistic in Proposition \ref{thm; GIC independent}.
Detection of unobservable FDI attacks is obtained as a by-product of the identification solution of Algorithm \ref{alg; identification}, $s$, 
where the proposed structural-constrained GIC-based detector decides that there is no attack if $s=0$
and that an unobservable FDI attack is present for the case where $s\ne0$.

\begin{algorithm}
	\caption{Structural-constrained GIC}
	\label{alg; identification}
	\begin{algorithmic}[1]
		\renewcommand{\algorithmicrequire}{\textbf{Input:}}
		\renewcommand{\algorithmicensure}{\textbf{Output:}}
		\REQUIRE  ~\\
		- difference-based measurements: $\Delta\zvec$ \\
		- network parameters: $\Hmat$, $\mathcal{L}$ \\
		- set of candidate state attack supports: $\mathcal{G}_{K_{\cvec}}$ \\
		- GIC penalty function: $\tau(\cdot, \cdot)$ 
		\ENSURE  selected hypothesis: $s$
		\FOR {$\Lambda_i\in\mathcal{G}_{K_{\cvec}}$ }
		\STATE evaluate $ \text{GIC}(\Lambda_i,\tau(\lvert\Lambda_i\rvert, \lvert\mathcal{L}\lvert) )$ from \eqref{eq; GIC} 
		\ENDFOR
		\RETURN $s= \arg \underset{i\in \{0,\ldots,\lvert\mathcal{G}_{K_{\cvec}}\rvert \}}{\max} ~~\text{GIC}(\Lambda_i,\tau(\lvert\Lambda_i\rvert, \lvert\mathcal{L}\lvert) )$
	\end{algorithmic} 
\end{algorithm}

The computational complexity of Algorithm \ref{alg; identification} is dominated by Step 2 in the \textbf{for} loop between Steps 1-3, in which 
the GIC statistic in \eqref{eq; GIC} is computed for each $\Lambda_i\in\mathcal{G}_{K_{\cvec}}$, 
where according to Assumption \ref{A; subsets}
\be \label{eq; computation}
\lvert \mathcal{G}_{K_{\cvec}} \rvert =\sum_{k=0}^{K_\cvec} \binom{\lvert \mathcal{V} \rvert} {k},
\ee
and $K_{\cvec}$ is the maximal sparsity level from \eqref{eq; G K C}. 
For any $ k=1,2,\ldots ,K_{\cvec},$ the binomial coefficient, $ {\binom{\lvert \mathcal{V} \rvert }{k}}$, grows by $O(\lvert \mathcal{V} \rvert^k)$ (p. 36, in \cite{elad2010sparse}). 
Thus, the number of times the GIC statistic is evaluated in Algorithm \ref{alg; identification}, which is the number of elements in \eqref{eq; computation}, is in the order of $O(\sum_{k=0}^{K_{\cvec}} \lvert \mathcal{V} \rvert^k )$.
The number of matrix-vector multiplications required for the computation of each GIC statistic in \eqref{eq; GIC} is $O(\lvert \mathcal{L}\rvert \lvert \Lambda_i \rvert^2 +(\lvert \Lambda_i \rvert+1)\lvert \mathcal{L}\rvert^2+\lvert\mathcal{L}\rvert )$.
Thus, for large-scale power systems, where $\lvert \mathcal{V}\rvert$ and $\lvert \mathcal{L}\rvert$ are significantly large compared with $K_{\cvec}$, 
the complexity is in the order of $O(\lvert \mathcal{V} \rvert^{K_{\cvec}}(K_{\cvec}+1) \lvert \mathcal{L}\rvert^2)$, and thus it may be infeasible to use the structural-constrained GIC method. 
In response to this issue, in Section \ref{sec; large scale}, we provide two low-complexity approximations.

\subsection{Performance analysis of the  GIC method}   \label{sec; performance}
In this section, we provide a theoretical  performance  analysis of the  GIC method.
In particular, we  derive:  1) the oracle GLRT detector \cite{elad2010sparse} that has access to the support of the sparse vector, $\Lambda_i\in\mathcal{G}_{k_{\cvec}}$; and
2) the clairvoyant detector  \cite{kayfundamentals}, which is a GLRT detector that has full access to the nuisance parameters.
While both the oracle and the clairvoyant detectors are infeasible since they are based on unavailable information, 
they provide insights insights with regard to the influence of the approximation error, $\epsilon_i \eta$, in the general case, where the support is unknown.  

In \cite{Stoica_Selen_Li}  it is shown that the GIC test provides the same decision rule as the generalized likelihood ratio test (GLRT) for a binary hypothesis testing problem with a fixed threshold. 
Hence, by using the GIC test in \eqref{eq; GIC}, the GLRT for the associated binary (modified) hypothesis
testing problem in \eqref{eq; model selection problem} with only $\Hnull$ and a single $\mathcal{H}_i$ is given by 
\be \label{eq; GLRT binary 2}
T^{GLRT|i}=\frac{1}{\sigma_{\evec}^2}	\lvert \lvert \Pmat_{\Lambda_i}   \Delta\zvec_{\mathcal{L}} \rvert\rvert^2\mathop{\gtrless}_{\Hnull}^{\mathcal{H}_i} \gamma^{GLRT},
\ee
where $\gamma^{GLRT}$ is the oracle detector threshold,
which is set to achieve a desired false alarm probability.
The following proposition states upper and lower bounds on the probability of false alarm, $P_{fa}$, in detecting unobservable FDI attacks with a given support, $\Lambda_i$, 
for the state attack vector, $\cvec$.  
\begin{prop}\label{prop; discussion}
	The probability of false alarm of the GLRT in \eqref{eq; GLRT binary 2} satisfies the following inequality:
	\be \label{eq; fa bound}
	\mathcal{Q}_{\frac{\lvert\Lambda_i\rvert}{2}}  \left( 0,\sqrt{\gamma^{GLRT} }\right)\le P_{fa}\le \mathcal{Q}_{\frac{\lvert\Lambda_i\rvert}{2}}  \left( \frac{\sqrt{\eta}}{\sigma_e},\sqrt{\gamma^{GLRT} }\right),
	\ee
	where $\mathcal{Q}_{v}(a,b)$ is the generalized Marcum $\mathcal{Q}$-function of order $v$ \cite{sun2010monotonicity}.
\end{prop} 
\begin{proof}
	The proof appears in Appendix \ref{App; discussion}.
\end{proof}

The upper and lower bounds in \eqref{eq; fa bound}  describe the influence of $\eta$  and $\gamma^{GLRT}$ on the probability of false alarm of the GLRT detector in \eqref{eq; GLRT binary 2} and 
are dictated by the property $0\le \epsilon_i \le 1$ shown in Appendix
\ref{App; discussion}. 
Since $\mathcal{Q}_{v}(a,b)$ decreases as $b$ decreases \cite{sun2010monotonicity},
we obtain the expected result that both the upper and lower bounds of $P_{fa}$ decrease as the threshold $\gamma^{GLRT}$ decreases. 
In contrast, since the function $\mathcal{Q}_{v}(a,b)$ increases as $a$ increases \cite{sun2010monotonicity}, 
the upper bound on the r.h.s. of \eqref{eq; fa bound} increases as $\eta$ increases.
This result indicates that when the bound on the nuisance parameter, $\Hmat\Delta\bt$, in \eqref{eq; quasi-static} is tighter, 
we can guarantee a tighter upper bound on the probability of false alarm, $P_{fa}$. Moreover,
under Assumption \ref{A; loads},
$\epsilon_i$ from \eqref{eq; approximation error} is  expected to be small for all $i=1,2,\ldots, \lvert\mathcal{G}_{K_{\cvec}}\rvert$. 
Thus, following the derivations in Appendix \ref{App; discussion}, the probability of false alarm is expected to be close to the lower bound in \eqref{eq; fa bound}. 
In particular,  the extreme case of ${\epsilon_i=0}$ occurs 
when the load changes are outside of the column space $\text{col}(\Hmat_{\mathcal{L},\Lambda_i})$, i.e.  $\Pmat_{\Lambda_i}\Hmat_{\mathcal{L},\mathcal{V}}\Delta\bt=\zerovec$.
For this case, the probability of false alarm achieves the lower bound in \eqref{eq; fa bound}, 
i.e. 
$$P_{fa}=\mathcal{Q}_{\frac{\lvert\Lambda_i\rvert}{2}}  \left( 0,\sqrt{\gamma^{GLRT}} \right), $$
and the probability of detection (see \eqref{eq; detect} in Appendix \ref{App; discussion}) is 
\be \label{eq; detect 0}
P_{d}=\mathcal{Q}_{\frac{\lvert\Lambda_i\rvert}{2}}   \left(\frac{\lvert\lvert \Hmat_{\mathcal{L},\Lambda_i}\cvec_{\Lambda_i}   \rvert\rvert}{\sigma_e},\sqrt{\gamma^{GLRT} }\right).
\ee
Hence, following the discussion on the properties of the generalized Marcum $\mathcal{Q}$ function in Appendix \ref{App; discussion}, the probability of detection in \eqref{eq; detect 0} increases as the attack energy, $\lvert\lvert \Hmat_{\mathcal{L},\Lambda_i}\cvec_{\Lambda_i}   \rvert\rvert^2$, increases.

In the following, we analyze the performance of the associated clairvoyant detector. 
The clairvoyant detector in the general case is the detector obtained by using the otherwise unknown signal parameter values of the composite hypothesis testing 
(see, e.g., Chapter 6.5, \cite{kayfundamentals} and \cite{zayyani2016double}). 
Therefore, the clairvoyant detector can serve as a reference for the achievable performance of practical detectors. 
In the considered model, the clairvoyant detector assumes perfect knowledge regarding the nuisance parameter, $\Hmat_{\mathcal{L},\mathcal{V}}\Delta\bt$.
Thus, similar to the GLRT in \eqref{eq; GLRT binary 2},
the clairvoyant GLRT 
is the GLRT detector for the binary (modified) hypothesis
testing problem in \eqref{eq; model selection problem} with only $\Hnull$ and a single $\mathcal{H}_i$, but now when $\Hmat\Delta\bt$ is assumed to be known. 
That is, the clairvoyant GLRT is given by
\be \label{eq; GLRT binary 3}
T^{GLRT|i}_{\normalfont{ cl }}=	\frac{1}{\sigma_{\evec}^2}	\lvert\lvert\Pmat_{\Lambda_i}(\Delta\zvec_{\mathcal{L}}-\Hmat_{\mathcal{L},\mathcal{V}} \Delta\bt)\rvert\rvert^2 \mathop{\gtrless}_{\Hnull}^{\mathcal{H}_i} \gamma^{\normalfont{ cl}},
\ee
where $\gamma^{\normalfont{cl }}$ is the clairvoyant detector threshold.
It should be noted that the clairvoyant detector is a theoretical detector that provides an upper bound on the performance of the probability of detection of practical detectors.

The following proposition presents a special case in which the proposed GLRT in \eqref{eq; GLRT binary 2} achieves the clairvoyant GLRT.
\begin{prop}\label{prop; oracle}
	If ${\epsilon_i=0}$,
	then the GLRT in \eqref{eq; GLRT binary 2} coincides with the clairvoyant GLRT in \eqref{eq; GLRT binary 3}.
\end{prop}
\begin{proof}
	In Appendix \ref{App; discussion} it is shown that $\epsilon_i=0$ implies that $\lvert\lvert\Pmat_{\Lambda_i}\Hmat_{\mathcal{L},\mathcal{V}} \Delta\bt\rvert\rvert^2 =\zerovec$. 
	Hence, as a result of known norm properties we obtain that $\Pmat_{\Lambda_i}\Hmat_{\mathcal{L},\mathcal{V}} \Delta\bt=\zerovec$. 
	Thus, by substituting this result in the clairvoyant GLRT in \eqref{eq; GLRT binary 3} we obtain the proposed GLRT in \eqref{eq; GLRT binary 2}. 
\end{proof}

As a result of Proposition \ref{prop; oracle}, if $\epsilon_i=0$, then the  detection of the attack is not influenced by the presence of  the nuisance parameter, $\Hmat\Delta\bt$. 
In other words, the value of $\Hmat\Delta\bt$ does not affect the performance of the proposed GLRT detector in \eqref{eq; GLRT binary 2} for the observed case. 

Similar to the clairvoyant GLRT, we define the clairvoyant GIC as the GIC selection rule for the hypothesis testing in \eqref{eq; model selection problem the}
that assumes perfect knowledge of the nuisance parameter, $\Hmat_{\mathcal{L},\mathcal{V}}\Delta \bt$.
Therefore, similar to the derivation of \eqref{eq; GIC}, it can be shown that the clairvoyant GIC is given by
\be \label{eq; GIC oracle}
\begin{aligned}
	\text{\normalfont{GIC}}_{\normalfont{ cl }}(\Lambda_i,\tau(\lvert\Lambda_i\rvert, \lvert\mathcal{L}\lvert)  )&= \frac{1}{\sigma_{\evec}^2}\lvert\lvert \Pmat_{\Lambda_i}(\Delta\zvec_{\mathcal{L}} -\Hmat\Delta \bt) \rvert\rvert^2 \\ 
	&~~~~-\tau(\lvert\Lambda_i\rvert, \lvert\mathcal{L}\lvert)+const.
\end{aligned}
\ee
The performance of the clairvoyant GIC is used as a benchmark on the performance of the proposed methods, as demonstrated in the simulations.

\section{Low-complexity identification methods}
\label{sec; large scale}
As shown at the end of  Subsection \ref{sec; detection and identification}, the computational complexity of the structural-constrained GIC method 
makes it impractical for large power system networks, 
where $\lvert \mathcal{V}\rvert$ and $\lvert\mathcal{L}\rvert$ are large. 
In this section, we develop two low-complexity methods that rely on Assumptions \ref{A; diff}-\ref{A; loads}: 
1) an OMP-based method \cite{mallat1993matching};
and 2) a novel method that exploits the graph Markovian property of order two between the nodes (buses) in the graph representation of the power system \cite{sedghi2015statistical}. 

\subsection{OMP method}
\label{sec; OMP}
The OMP algorithm in \cite{mallat1993matching} is an efficient method for the recovery of sparse signals. 
The basic principle behind the OMP algorithm is to iteratively find the support set of the sparse vector.
The OMP method proceeds by finding the column of the CS matrix
that correlates most strongly with the signal residual.
The residual is constructed in each iteration by projecting the measurements onto the linear space spanned by the remaining columns that were not selected in previous iterations. 

In this subsection, we apply the OMP algorithm for the sparse recovery of the state attack vector, $\cvec$, 
which is a sparse signal as described in Assumption \ref{A; subsets}, from the measurements in \eqref{eq; 2}.
It should be noted that the measurement model in \eqref{eq; 2} contains a nuisance parameter vector, $\Hmat_{\mathcal{L},\mathcal{V}}\Delta\bt$,
which is not a part of the conventional sparse recovery model.
Nevertheless, similarly to in the derivations in Subsection \ref{sec; framework}, it can be shown that under Assumptions \ref{A; diff}-\ref{A; loads}
the nuisance parameter has a negligible effect on the OMP selection criterion.
Thus, the conventional OMP method is valid for this setting. 

The main iteration of the OMP algorithm, which is performed on the load measurements model from \eqref{eq; 2},
is given as follows. 
Suppose $\Lambda^{(j)}$ is the estimated support set of $\cvec$ in the $j$th iterative step.
In the $j+1$th iteration, we compute
\be \label{eq; main OMP}
k=\arg \underset{ \tilde{k}\in \mathcal{V}} {\max} ~~\lvert\lvert \Pmat_{\{  \tilde{k} \}  } \rvec^{\scriptstyle{\text{OMP}} |(j)}  \rvert\rvert^2,
\ee
where $\Pmat_{\tilde{k}}$ is obtained by 
replacing $\Lambda_i=\tilde{k}$ in \eqref{eq; projection matrix},
and 
\be \label{eq; signal residual 2}
\rvec^{\scriptstyle{\text{OMP} }|(j)}=\Pmat_{\Lambda^{(j)}}^{\bot}\Delta\zvec_{\mathcal{L}}
\ee
is the signal residual at the $j$th iteration.
By substituting \eqref{eq; signal residual 2} and $\Pmat_{\Lambda^{(j)}}^{\bot}= \Imat-\Pmat_{\Lambda^{(j)}}$ on the r.h.s. of \eqref{eq; main OMP} we obtain 
\be \label{eq; residual decompose 1}
\lvert\lvert \Pmat_{\{k\}} \rvec^{\scriptstyle{\text{OMP}}|(j)}  \rvert\rvert^2
= \lvert\lvert \Pmat_{\{k\}}\Delta\zvec_{\mathcal{L}} -\Pmat_{\{k\}}\Pmat_{\Lambda^{(j)}}\Delta\zvec_{\mathcal{L}} \rvert\rvert^2. \\
\ee
By limiting the number of iterations to the maximal sparsity level, $K_{\cvec}$, we obtain that $\Lambda^{(j)}\in\mathcal{G}_{K_{\cvec}}$.
In addition, any single bus, $k\in\mathcal{V}$, is an element in $\mathcal{G}_{K_{\cvec}}$, i.e. $k\in\mathcal{G}_{K_{\cvec}}$, as well. 
Therefore, under Assumptions \ref{A; diff}-\ref{A; loads}, we obtain from \eqref{eq; projection loads} that 
the nuisance parameter,  $\Hmat_{\mathcal{L},\mathcal{V}}\Delta\bt$, has a minor effect on both the terms $\Pmat_{\{k\}}\Delta\zvec_{\mathcal{L}}$ and 
$\Pmat_{\Lambda^{(j)}}\Delta\zvec_{\mathcal{L}}$. 
Consequently, the nuisance parameter has a minor effect on \eqref{eq; residual decompose 1} and, thus, its influence on the OMP selection procedure in \eqref{eq; main OMP} can be neglected.
The proposed structural-constrained OMP algorithm for the identification of unobservable FDI attacks is provided in Algorithm \ref{alg; OMP}.

\begin{algorithm}
	\caption{Structural-constrained OMP}
	\label{alg; OMP}
	\begin{algorithmic}[1]
		\renewcommand{\algorithmicrequire}{\textbf{Input:}}
		\renewcommand{\algorithmicensure}{\textbf{Output:}}
		\REQUIRE ~\\
		- difference-based measurements: $\Delta\zvec$ \\
		- network parameters: $\Hmat$, $\mathcal{V}$, $\mathcal{L}$ \\
		- maximal sparsity level: $K_{\cvec}$ \\
		- stopping rule threshold: $\gamma^{\scriptstyle{\text{OMP}}} $
		\ENSURE  estimated support: $\hat{\Lambda}^{\scriptstyle{\text{OMP}}}$ 
		\\ \textit{Initialize} : $\Lambda^{(0)}=\emptyset,~ \rvec^{\scriptstyle{\text{OMP}}|(0)} =\Delta\zvec_{\mathcal{L}}$ 
		\FOR {$j=1\ldots K_{\cvec}$}
		\STATE evaluate $k$ from \eqref{eq; main OMP}
		\IF {$\lvert\lvert \Pmat_{\{k\}}\rvec^{\scriptstyle{\text{OMP}}|(j-1)}\rvert\rvert^2<\gamma^{\scriptstyle{\text{OMP}}}$}
		\STATE \textbf{return} $\Lambda_s=\hat{\Lambda}^{(j-1)}$
		\ENDIF
		\STATE update: $\Lambda^{(j)}=\{\Lambda^{(j-1)}\cup k\}$ 
		\STATE update: $\rvec^{\scriptstyle{\text{OMP}}|(j)}=\Pmat_{\Lambda^{(j)}}^{\bot}\Delta\zvec_{\mathcal{L}}$ 
		\ENDFOR 
		\RETURN $\hat{\Lambda}^{\scriptstyle{\text{OMP}}}=\hat{\Lambda}^{(j)}$ 
	\end{algorithmic} 
\end{algorithm}

The computational complexity of Algorithm \ref{alg; OMP} is dominated by Step 3 (computing \eqref{eq; main OMP}), 
which requires $O(2\lvert\mathcal{L}\rvert^2+3\lvert \mathcal{L}\rvert)$ matrix-vector multiplications for each $k=1,2,\ldots,\lvert\mathcal{V}\rvert$.
The loop between Steps 2-8 is performed at most $K_{\cvec}$ times.
As a result, the total complexity of Algorithm \ref{alg; OMP}
is in the order of $O(\lvert\mathcal{V}\rvert K_c\lvert\mathcal{L}\rvert^2)$, 
which is significantly lower 
than the complexity of the GIC method (see discussion after \eqref{eq; computation}).

In general, the OMP method is used in a variety of applications due to its low computational complexity. 
However, the OMP method is a greedy algorithm with no optimal recovery guarantees,
and usually requires an incoherent dictionary in order to provide high performance \cite{Tropp_2004}.
In our case, the CS matrix is the topology matrix $\Hmat_{\mathcal{L},\mathcal{V}}$, 
which may be highly correlated, and thus, with large mutual coherence. 
In order to address this issue, in the following subsection, 
we develop a novel low-complexity method
that uses the power system graph representation Markovian properties.

\subsection{Graph Markovian GIC (GM-GIC)}
\label{sec; GM-GIC}

In this section, we develop the low complexity GM-GIC method 
for the measurement model in \eqref{eq; 2} and under Assumptions \ref{A; diff}-\ref{A; loads}, 
which considers the graphical representation of the power system. 
Accordingly, in this subsection the system buses are referred to as the graph nodes.
The power system graphical representation is utilized in Subsection \ref{sec; attack analysis}
to analyze the affect of an unobservable FDI attack on the measurements. 
Based on this analysis,
we derive the GM-GIC method, which consists of the following four stages:
1) low-scale pre-screening (Subsection \ref{sec; pre screening});  2) node partitioning (Subsection \ref{sec; node partition}); 
3) local GIC, on the partitioned subsets; 
and 4) sparsity correction.
The GM-GIC method, including Stages 3 and 4, is summarized in Subsection \ref{sec; GM-GIC small}.
Furthermore, in Subsection \ref{sec; GSP extensions} it is demonstrated that the GM-GIC method 
can be applied in GSP applications with sparse signals,
in addition to its application for FDI attack identification in power systems.

\subsubsection{Attack analysis} \label{sec; attack analysis}
In the following, we analyze the effect of an unobservable FDI attack on the system measurements. 
Specifically, we focus on how an attack, with a support $\Lambda$, affects the single-node measurement subspace, 
$\text{col}(\Hmat_{\mathcal{L},m})$, for any node $m\in\mathcal{V}$.
Considering the measurement model in \eqref{eq; 2}, the unobservable FDI attack can be linearly decomposed as follows:
\be \label{eq; GM attack}
\begin{aligned}
	\Hmat\cvec=	\Hmat_{\mathcal{L},\Lambda}\cvec_{\Lambda}=\sum\limits_{k\in\Lambda} \Hmat_{\mathcal{L},k}c_k.
\end{aligned}
\ee
Based on the decomposition in \eqref{eq; GM attack}, we can analyze the influence of an attack on a single-node measurement subspace 
by summing over the individual influences on each node, $k\in\Lambda$.
The following proposition evaluates the effect of a single-node attack on a single-node measurement subspace.
\begin{prop}\label{prop; single attack}
	The projection of a single-node attack on node $k\in\Lambda$, $\Hmat_{\mathcal{L},k}c_k$, 
	onto the measurement subspace associated with node $m\in\mathcal{V}$, $\text{col}(\Hmat_{\mathcal{L},m})$,
	satisfies
	\be \label{eq; 2nd 0}
	\Pmat_{\{m\}}\Hmat_{\mathcal{L},k}c_k=\boldsymbol{0},~~\forall k,m\in\mathcal{V},~d(k,m)>2
	\ee
	and
	\be \label{eq; all orders}
	\begin{aligned}
		\lvert\lvert\Pmat_{\{m\}}\Hmat_{\mathcal{L},k}c_k\rvert\rvert^2\le \lvert\lvert\Pmat_{\{k\}}\Hmat_{\mathcal{L},k}c_k\rvert\rvert^2,~~\forall k,m\in\mathcal{V},
	\end{aligned}
	\ee
	where $\Pmat_{\{m\}}$ is obtained by substituting $\Lambda=\{m\}$ in \eqref{eq; projection matrix}.
\end{prop}
\begin{proof}
	The proof appears in Appendix \ref{App; single attack}.
\end{proof}

This proposition demonstrates that the single-node attack obtains a `local' effect.
That is, a single-node attack on node $k$ does not affect the single-node measurement subspace of node $m$ if the 
geodesic distance (hop distance) \cite{sedghi2014conditional} between nodes $k$ and $m$, $d(k,m)$, is greater than two, $d(k,m)>2$. 
By multiplying \eqref{eq; GM attack} by $\Pmat_{\{m\}}$ and  substituting \eqref{eq; 2nd 0}, we obtain 
\be \label{eq; 0 is 2}
\Pmat_{\{m\}}\Hmat_{\mathcal{L},\Lambda}c_{\Lambda}=\zerovec,~~\forall k\in\Lambda,~d(k,m)>2. 
\ee 
As a result of \eqref{eq; 0 is 2}, the single-node measurement subspaces affected by the unobservable FDI attack are only 
those which are associated with nodes in the set 
\be \label{eq; A set}
\mathcal{A}\define\{m\in\mathcal{V}:~\exists k\in\Lambda~\text{s.t.}~0\le d(k,m)\le 2\},
\ee
where, according to \eqref{eq; A set}, the set $\mathcal {A}$ includes the attacked nodes, 
those in $\Lambda$, and first- or second-order neighbors of attacked nodes. 
This observation can be interpreted as a second-order graph Markov property \cite{sedghi2015statistical}. 

In general, the power system network generates a \textit{sparse} graph where nodes are connected to 2-5 neighbors \cite{wang2010generating}.
Thus, the nodes in the power system have at most 5 first-order and $25$ second-order neighbor nodes. 
As a result, the number of nodes affected by the attack is bounded:
\be \label{eq; A bound}
\lvert \mathcal{A}\rvert \le  (1+5+25) K_{\cvec},
\ee
where $K_{\cvec}$ is the sparsity term defined in Assumption \ref{A; subsets} and $\mathcal{A}$ is defined in \eqref{eq; A set}.
Moreover, the power system network diplays a local behavior for the connectivity pattern, where only substations geographically close are likely to be connected. 
Thus, the bound in \eqref{eq; A bound} is not a tight bound, and $\lvert \mathcal{A}\rvert $ is expected to be significantly lower than the r.h.s of \eqref{eq; A bound}.
In conclusion, it is implied that for large networks $\lvert\mathcal{A}\rvert\ll\lvert\mathcal{V}\rvert$.

\subsubsection{ Pre-screening} \label{sec; pre screening}
The first stage of the GM-GIC method is a pre-screening stage. 
In order to evaluate if the dictionary matrix column associated with node $m\in\mathcal{V}$, $\Hmat_{\mathcal{L},m}$, is correlated with the attack, $\Hmat\cvec$, 
we derive the GLRT (in a similar manner to in \eqref{eq; GLRT binary 2}) for the 
associated binary hypothesis testing problem in \eqref{eq; model selection problem} 
with only $\Hnull$ and a single $\mathcal{H}_i$, selected as $\mathcal {H}_i=\{m\}$, which results in
\be \label{eq; GLRT decision rule}
\lvert \lvert \Pmat_{\{m\}}   \Delta\zvec_{\mathcal{L}} \rvert\rvert^2  \mathop{\gtrless}_{\Hnull}^{\mathcal{H}_1} \rho, 
\ee
where $\rho$ is set to determine a desired false alarm rate. 
Applying the GLRT detector in \eqref{eq; GLRT decision rule} for any node in the system
yields the following set of suspicious nodes: 
\be  \label{eq; suspected}
\mathcal{S}\define\{m\in\mathcal{V}:~\lvert \lvert \Pmat_{\{m\}}   \Delta\zvec_{\mathcal{L}} \rvert\rvert^2  >\rho\}.
\ee
Thus, the node $m$ will be included in the set $\mathcal{S}$ if abnormal energy is detected in the single-node measurement subspace, $\text{col}(\Hmat_{\bcolon,m})$. 
By considering the hypothesis testing in \eqref{eq; model selection problem} in which $\Lambda_i=\{m\}$, the projection of the nuisance parameter, $\Pmat_{\{m\}}\Hmat_{\mathcal{L},\mathcal{V}}\Delta\bt$, is negligible. 
Therefore, the inclusion of the node $m$ in the set $\mathcal{S}$ indicates that the measurement subspace, $\text{col}(\Hmat_{\bcolon,m})$, 
is affected by the attack, i.e. that $\Pmat_{\{m\}}\Hmat_{\mathcal{L},\Lambda}\cvec_{\Lambda}\ne \zerovec$.
Thus, from \eqref{eq; 0 is 2}, the set of suspicious nodes in \eqref{eq; suspected}   
can be interpreted as an estimator of the set $\mathcal{A}$ in \eqref{eq; A set}.
From \eqref{eq; all orders}, it can be seen that the norm of the projected single-node attack on single-node measurement subspaces of attacked nodes 
is higher than the norm of the projection onto measurement subspaces associated with other nodes. 
Thus, considering the discussion after \eqref{eq; GM attack},  the effect of the attack, $\Hmat_{\mathcal{L},\Lambda}\cvec_{\Lambda}$, on attacked nodes is expected to be higher than 
its effect on their first- and second- order neighbor nodes.
Therefore, even if the estimator $\mathcal{S}$ does not cover all the nodes in $\mathcal{A}$, 
the attacked nodes are still expected to be included.

\subsubsection{Node partitioning} \label{sec; node partition}
The second stage of the GM-GIC method is the partitioning of the suspicious set, $\mathcal{S}$, in \eqref{eq; suspected}, 
into disjoint subsets.
This partitioning provides the condition that an attack on nodes located in one of the subsets does not affect the measurement subspaces 
associated with nodes in the other subsets. 
According to the proposed partitioning, 
a general set $\mathcal{S}$ is partitioned into $Q$ disjoint subsets, $\{\mathcal{S}_q\}_{q=1}^Q$, 
if the following is satisfied:
\be \label{eq; different subsets}
d(k,m)>2, ~\forall{ k\in\mathcal{S}_{q},m\in\mathcal{S}_{p}}
\ee
for any two different subsets, $S_q$ and $S_p$, $p\neq q$, selected from $\{\mathcal{S}_q\}_{q=1}^Q$,  where $1\leq Q \leq K_c$.
Thus, from \eqref{eq; 0 is 2} and \eqref{eq; different subsets}, an attack on any node subset in $\mathcal{S}_{q}$, $\Lambda_{q}\in\mathcal{S}_{q}$, 
does not affect the column space associated with the nodes in any of the other subsets, e.g. $\text{col}(\Hmat_{\mathcal{L},\mathcal{S}_{p}})$, i.e. 
\be \label{eq; no affect}
\Pmat_{\mathcal{S}_{p}}\Hmat_{\mathcal{L},\Lambda_{q}}=\zerovec.
\ee
Therefore, a node partitioning that satisfies \eqref{eq; different subsets} enables identification of the attacked buses in each subset separately,
and, thus, reduces the problem dimension. 
In contrast, the GIC from Algorithm \ref{alg; identification} considers the entire node set $\mathcal{V}$.

In practice, finding a node partitioning that satisfies \eqref{eq; different subsets} is performed as follows. 
First, based on the graph representation of the given power system, we generate the unweighted undirected graph $\mathcal{\tilde G}=(\mathcal{S}, \mathcal{E}_{\mathcal{\tilde G}}  )$, 
where $\mathcal{S}$ are the nodes and the set of edges is defined as
\be \label{eq; graph expand edges}
\mathcal{E}_{\mathcal{\tilde G}}\define \{(k,m):~k,m\in\mathcal{S},~ 1\le d(k,m) \le 2\},
\ee
in which $d(\cdot,\cdot)$ refers to the geodesic distance measured on the original graph that represents the power system network. 
Then, we find the connected components of $\tilde{\mathcal{G}}$ 
(e.g. by using the Matlab command conncomp). 
According to the definition in \eqref{eq; graph expand edges}, 
we obtain that $d(k,m)>2$ for any $k$ and $m$ that belong to different connected components of $\tilde{\mathcal{G}}$. 
Therefore, selecting the partition subsets $\{\mathcal{S}\}_{q=1}^Q$ to be the node sets of the connected components of $\tilde{G}$
satisfies \eqref{eq; different subsets}.

\subsubsection{Summary: GM-GIC method} \label{sec; GM-GIC small}
In this subsection we summarize the proposed GM-GIC method.
In the first stage of this method, a reduced set, composed of suspicious nodes, $\mathcal{S}$, is extracted from the node set $\mathcal{V}$, according to \eqref{eq; suspected}. 
In the second stage, the set $\mathcal{S}$ is partitioned into $Q$ disjoint subsets, $\{\mathcal{S}_q\}_{q=1}^Q$, as described in Subsection \ref{sec; node partition}.
In the third stage, the GIC method is applied on each subset separately, by replacing the set of candidate state attack supports,
$\mathcal{G}_{K_{\cvec}}$, in Algorithm \ref{alg; identification} with 
\be  \label{eq; G K C q}
\mathcal{G}_{K_{\cvec}}^{\hspace{0.05cm}\mathcal{S}_q}\define\{\Lambda\subset \mathcal{S}_q:~1\le \lvert \Lambda \rvert  \le K_{\cvec}  \}.
\ee
As a result, for each subset $\mathcal{S}_q$ we obtain a partial estimation of the support set, which is denoted by $\hat\Lambda_{q}$.
The total support set of the state attack vector is the union of all the partial estimates
\be \label{eq; joint estimated support}
\hat{\Lambda}^{\scriptstyle{\text{GM-GIC}}}_{\text{temp}}=\bigcup\limits_{q=1}^{Q}\hat{\Lambda}_q.
\ee

The estimated support in \eqref{eq; joint estimated support} may exceed the sparsity limit $\lvert \hat{\Lambda}^{\scriptstyle{\text{GM-GIC}}}_{\text{temp}}\rvert \le K_{\cvec}$.
Hence, $\hat{\Lambda}^{\scriptstyle{\text{GM-GIC}}}_{\text{temp}}$ may include node-elements which are not attacked.
It should be noted that in this stage, Assumption \ref{A; restricted} is relaxed and the sparsity restriction is only imposed on the partial (separated) estimates, i.e. $\Lambda_q\le K_{\cvec}$.
In order to reduce the identification errors that may be induced by this relaxation,
in the fourth stage, the estimated support in \eqref{eq; joint estimated support} is corrected to satisfy $\lvert\hat{\Lambda}^{\scriptstyle{\text{GM-GIC}}}\rvert\le K_{\cvec}$.
This stage is performed by: 1) evaluating the state attack ML estimation, $\hat{\cvec}^{\text{ML}|i}_{\hat{\Lambda}^{\scriptstyle{\text{GM-GIC}}}}$,
by replacing $\Lambda_i$ with \eqref{eq; joint estimated support} in \eqref{eq; estimated state attack vector 2} from Appendix \ref{App;  proof log likelihood}; 2) sorting $\hat{\cvec}^{\text{ML}|i}_{\hat{\Lambda}^{\scriptstyle{\text{GM-GIC}}}_{\text{temp}}}$ in a descending order:
\be \label{eq; I set}
\mathcal{I}\define \{i_1,i_2,\ldots, i_{\lvert\hat{\Lambda}^{\scriptstyle{\text{GM-GIC}}}\rvert} \}=\text{sort}(\hat{\cvec}^{\text{ML}|i}_{\hat{\Lambda}^{\scriptstyle{\text{GM-GIC}}}_{\text{temp}}});
\ee and 3) preserving only the $K_{\cvec}$ elements with the highest absolute values,
\be \label{eq; I select}
\hat{\Lambda}^{\scriptstyle{\text{GM-GIC}}}=\{i_1,i_2,\ldots, i_{K_{\cvec}}\}.
\ee
The proposed GM-GIC algorithm for the identification of unobservable FDI attacks is provided in Algorithm \ref{alg; large scale}.

\begin{algorithm}
	\caption{GM-GIC}
	\label{alg; large scale}
	\begin{algorithmic}[1]
		\renewcommand{\algorithmicrequire}{\textbf{Input:}}
		\renewcommand{\algorithmicensure}{\textbf{Output:}}
		\REQUIRE ~\\
		- difference-based measurements: $\Delta\zvec$ \\
		- network parameters: $\Hmat$, $\mathcal{V}$, $\mathcal{L}$ \\
		- maximal sparsity level: $K_{\cvec}$ \\
		- energy threshold: $\rho$ 
		\ENSURE estimated support: $\hat{\Lambda}^{\scriptstyle{\text{GM-GIC}}}$
		\STATE Pre-screening: evaluate $\mathcal{S}$ by \eqref{eq; suspected}
		\IF {$\mathcal{S}=\emptyset$}
		\STATE \textbf{return} $\hat{\Lambda}^{\scriptstyle{\text{GM-GIC}}}=\emptyset$  \\
		\ENDIF
		\STATE 	Node-partitioning: 1) generate $\mathcal{\tilde G}=(\mathcal{S}, \mathcal{E}_{\mathcal{\tilde G}}  )$ by computing \eqref{eq; graph expand edges}; 
		2) partition $\tilde{\mathcal{G}} $ into its connected components (e.g. by conncomp in Matlab); and 3) set $\{\mathcal{S}_q\}_{q=1}^{Q}$ to be these components 
		\FOR {$q=1\ldots Q$}
		\STATE generate the graph $\mathcal{G}_{K_{\cvec}}^{\hspace{0.05cm}\mathcal{S}_q}$ by \eqref{eq; G K C q} 
		\STATE compute $\hat{\Lambda}_q$ by applying Algorithm \ref{alg; identification} 
		with the input set of candidate state attack supports $\mathcal{G}_{K_{\cvec}}^{\hspace{0.05cm}\mathcal{S}_q}$
		\ENDFOR
		\STATE compute  $\hat{\Lambda}^{\scriptstyle{\text{GM-GIC}}}$ by \eqref{eq; joint estimated support} \\
		\IF {$\lvert\hat{\Lambda}^{\scriptstyle{\text{GM-GIC}}}\rvert>K_{\cvec}$}
		\STATE compute $\hat{\cvec}^{\text{ML}|i}_{\hat{\Lambda}^{\scriptstyle{\text{GM-GIC}}}}$ 
		by \eqref{eq; estimated state attack vector 2} 
		\STATE correct $\hat{\Lambda}^{\scriptstyle{\text{GM-GIC}}}$ by \eqref{eq; I set} and \eqref{eq; I select}
		\ENDIF
		\RETURN $\hat{\Lambda}^{\scriptstyle{\text{GM-GIC}}}$	
	\end{algorithmic} 
\end{algorithm}

The computational complexity of the different stages of the GM-GIC method in Algorithm \ref{alg; large scale} are as follows:
1) the pre-screening stage in Step 1 requires $O(\lvert \mathcal{V}\rvert(2\lvert \mathcal{L}\rvert^2 +3\lvert\mathcal{L}\rvert))$ matrix-vector multiplications;
2) the node partitioning stage in Step 3 is dominated by the graph partitioning, which is implemented by the conncomp($\cdot$) Matlab command and requires $O(\lvert \mathcal{S}\rvert+\lvert \mathcal{S}\rvert ^2)$ computations 
(considering that the number of edges in $\tilde{\mathcal{G}}$ cannot be anticipated in advance, this analysis is based on the worst case in which $\tilde{\mathcal{G}}$ is a fully connected graph); 
3) in Steps 4-7 the GIC method is applied on each of the subsets $\{\mathcal{S}_q\}_{q=1}^Q$, so that, 
from Subsection \ref{sec; detection and identification},
the complexity is in the order of 
$O((\sum_{q=1}^{Q} \lvert \mathcal{S}_q\rvert ^{K_{\cvec}})(K_{\cvec}+1)\lvert\mathcal{L}\rvert ^2 )$; and 
4) the sparsity correction stage, in Steps 9-12, is dominated by Step 10, which requires $O( K_{\cvec}(K_{\cvec}+1)\lvert\mathcal{L}\rvert)$ matrix-vector multiplications.
In power systems, the number of load nodes (buses) $\lvert \mathcal{L}\rvert$ is, in general, in scale with the number of nodes $\lvert \mathcal{V}\rvert$.
Thus, from Subsections \ref{sec; attack analysis}-\ref{sec; pre screening},  we obtain that $\lvert \mathcal{S}\rvert \ll\lvert \mathcal{L} \rvert$ and
the total computational complexity of Algorithm \ref{alg; large scale} is $O((\sum_{q=1}^{Q} \lvert \mathcal{S}_q\rvert ^{K_{\cvec}}) (K_{\cvec}+1)\lvert\mathcal{L}\rvert ^2 )$.

The computational complexity of the
different methods is summarized in Table \ref{t; complexity}.
It can be seen that the computational complexity of the GM-GIC method is significantly lower than that of the complexity of the GIC method,
but may be higher than that of the complexity of the OMP method. 
Finally, the worst complexity for the GM-GIC method is obtained when $\mathcal {S}$ cannot be partitioned.
In this case, the complexity is in the order of $O(\lvert \mathcal{S}\rvert ^{K_{\cvec}} (K_{\cvec}+1)\lvert\mathcal{L}\rvert ^2 )$, 
which is still significantly lower than the computational complexity of the GIC method.

In contrast to the GIC and the OMP methods, the computational complexity of the GM-GIC method also depends on the underlying structure of the power system (and not only the dimensions), 
which is represented by the topology matrix $\Hmat$. 
In particular, the results of the pre-screening and node partitioning stages, $\mathcal{S}$ and $\{\mathcal{S}_q\}_q$, which are detailed in Subsections \ref{sec; pre screening} and \ref{sec; node partition}, respectively, 
depend on the formation of the topology matrix $\Hmat$.  
The method will perform effectively when the properties detailed at the end of Subsection \ref{sec; attack analysis} are satisfied.

\begin{table}
	\begin{center}
		\begin{tabular}{ |c|c| } 
			\hline
			GIC & $O(\lvert \mathcal{V} \rvert^{K_{\cvec}}(K_{\cvec}+1) \lvert \mathcal{L}\rvert^2)$  \\
			\hline
			OMP & $O(\lvert\mathcal{V}\rvert K_c\lvert\mathcal{L}\rvert^2)$   \\
			\hline
			GM-GIC  & $O((\sum_{q=1}^{Q} \lvert \mathcal{S}_q\rvert ^{K_{\cvec}}) (K_{\cvec}+1)\lvert\mathcal{L}\rvert ^2 )$  \\
			\hline
			GM-GIC (worst case)   & $O(\lvert \mathcal{S}\rvert ^{K_{\cvec}} (K_{\cvec}+1)\lvert\mathcal{L}\rvert ^2 )$  \\
			\hline
		\end{tabular}
	\end{center}
	\caption{ Computational complexity of the proposed  methods.}
	\label{t; complexity}
\end{table}

\subsubsection{Sparse signal recovery in general GSP applications} 
\label{sec; GSP extensions}

The proposed GM-GIC heuristic is designed to exploit the graphical properties of power system networks, which are detailted in Subsection \ref{sec; attack analysis}
as part of the attack analysis, and include: 
1) graph Markovity; 2) graph sparsity; and 3) local graph connectivity behavior. 
Thus, the GM-GIC method can be utilized for sparse recovery in GSP frameworks,
where, if Properties 1)-3) are satisfied, the GM-GIC method has advantages, 
compared with conventional sparse recovery algorithms, 
in terms of computational complexity and accuracy of the sparse recovery. 

The emerging field of GSP provides new methodologies for the analysis of signals in applications 
with underlying relations that could be modeled by a graph \cite{shuman2013emerging,sandryhaila2013discrete,ortega2018graph}.
In particular, the propagation of a process that originates from a sparse input through the graph vertex domain is commonly modeled in the GSP literature as an output of a graph filter  \cite{segarra2016blind,xu2011compressive,ramirez2021graph}; this approach has various applications, such as locating the source of a disease \cite{newman2002spread,pinto2012locating} or monitoring anomalies in sensor networks \cite{schizas2009distributed}.  
The proposed GM-GIC method can be applied to this problem of sparse recovery in GSP models as follows.
Let us assume the graph $\mathcal{G}=(\mathcal{V},\mathcal{E})$ with the set of nodes (vertices) $\mathcal{V}$ and the set of edges $\mathcal{E}$. 
We consider the recovery of a sparse graph signal, ${\xvec \in \mathbb{R}^{\lvert \mathcal{V} \rvert }}$,  s.t. $\lvert\lvert \xvec \rvert\rvert _0 \ll\lvert \mathcal{V} \rvert$,
from the noisy output of a graph filter, $\Fmat\in\mathbb{R}^{\vert \mathcal{V}\rvert\times \lvert \mathcal{V} \rvert}$:
\be \label{eq; general}
\yvec=\Fmat\xvec+\evec, 
\ee
where $\evec\sim \mathcal{N}(\zerovec,\sigma^2_n)$ is the system noise,
and the graph filter $\Fmat$ is linear and shift-invariant \cite{shuman2013emerging,sandryhaila2013discrete,ortega2018graph,segarra2016blind}.     
Hence, $\Fmat$ is a polynomial in a graph shift operator (GSO), $\Smat\in\mathbb{R}^{\vert \mathcal{V}\rvert\times \lvert \mathcal{V} \rvert}$,
as follows \cite{shuman2013emerging,sandryhaila2013discrete,ortega2018graph}:
\be \label{eq; graph filter}
\Fmat=h_0\Imat+h_1\Smat+\ldots +h_{\Psi}\Smat^{\Psi},
\ee
where $h_0,\ldots,h_{\Psi}$ are the filter's coefficients and $\Psi\geq 1$ is the filter's order.
The GSO, $\Smat$, is defined as a matrix that satisfies
\be \label{eq; GSO}
S_{k,m}=0, ~ \text{if}~d(k,m)>1,
\ee 
$\forall k,m\in\mathcal{V}$, where $d(k,m)$ is the geodesic distance between nodes $k$ and $m$. 
In particular, it can be verified that \eqref{eq; GSO} implies
\be \label{eq; F prop}
F_{n,k}=0, ~\text{if}~ d(n,k)>\Psi, 
\ee
$\forall k,m\in\mathcal{V}$, which leads to 
\be \label{eq; F prop 2}
\Fmat_{\mathcal{V},m}^T\Fmat_{\mathcal{V},k}=0,~~\forall k,m\in\mathcal{V},~d(k,m)>2\Psi. 
\ee
It can be seen that the problem in \eqref{eq; model selection problem} is reduced to the problem in \eqref{eq; general} 
when the topology matrix $\Hmat$ is replaced with the graph filter $\Fmat$ and when the nuisance parameter is absent, $\Delta\bt=\zerovec$.

It can be seen that for a graph filter of order $\Psi=1$, \eqref{eq; F prop 2} is identical to the property in \eqref{eq; App 0}-\eqref{eq; App 1}
from Appendix \ref{App; single attack}, where in this case $\Fmat=\Hmat_{\mathcal{V},\mathcal{L}}$. 
The property in \eqref{eq; App 0}-\eqref{eq; App 1} is the basis for Proposition \ref{prop; single attack} that enables the derivation of the GM-GIC method. 
Thus, for this case, the GM-GIC method can be implemented (without any change) on \eqref{eq; general} for the recovery of the sparse graph signal, $\xvec$, 
which replaces the state attack vector, $\cvec$, and without nuisance parameters. 
For the general case, where $\Psi>1$, the GM-GIC method can be applied on \eqref{eq; general} after replacing the constraint in \eqref{eq; graph expand edges}
(that is part of the node partitioning stage in Algorithm \ref{alg; large scale}) with the constraint $1\le d(n,k)\le 2\Psi$. 
For this case, in which \eqref{eq; F prop 2} is considered instead of \eqref{eq; App 1}, the condition in \eqref{eq; 2nd 0} from Proposition \ref{prop; single attack} is changed to $d(k,m) >  2\Psi$.

\section{Simulations}
\label{sec; Simulation study}
In this section, the performance of the proposed methods is demonstrated for the tasks of detection and identification of unobservable FDI attacks.
The simulations are conducted on the IEEE-30 bus test case,
where the topology matrix and measurement data are extracted using the Matpower toolbox for Matlab \cite{matpower}.
The simulation set-up is described in Subsection \ref{sec; sim set up},
while the proposed methods' performance is demonstrated in Subsection \ref{sec; sim small scale}.

\subsection{Set-up}
\label{sec; sim set up}

\subsubsection{Measurements}
The simulation study is conducted on the difference-based model in \eqref{eq; 1}.
For the sake of simplicity of implementation, 
we ensure Assumption \ref{A; restricted}
by defining the set $\tilde{\mathcal{V}}$ to include only state variables that are related to load measurements,
and then restricting the support of the state attack vector to $\Lambda\subseteq\tilde{\mathcal{V}}$. 
In particular, for the IEEE-30 bus test case it can be verified that ${\tilde{\mathcal{V}}=\{14,16,17,18,19,20\}}$.
In each simulation, we set the cardinality of the state attack support vector to be $\lvert\Lambda\rvert=K_{\avec}$ and draw the support vector, $\Lambda$,
uniformly from the set $\{\Lambda\in\tilde{\mathcal{V}}:~ \lvert\Lambda\rvert=K_{\avec}\}$.
The attack values on the chosen support are randomly drawn from the uniform distribution over the interval $[-1,1]$.
Then, the values are scaled to obtain a desired value of attack norm, $\lvert\lvert \avec\rvert\rvert$. 

The difference-based state vector is obtained by first generating the state vectors at times $t$ and $t+1$, $\bt_t$ and $\bt_{t+1}$, respectively, 
and then subtracting $\bt_t$ from $\bt_{t+1}$ to obtain $\Delta\bt=\bt_{t+1}-\bt_{t}$. 
The state vector at time $t$, $\bt_{t}$, is set to the values in the IEEE test case \cite{matpower}.
Based on $\bt_t$ and considering Assumption \ref{A; loads}, 
the state vector at time $t+1$, $\bt_{t+1}$, is simulated 
by first updating the load measurements with random scaling:
\be \label{eq; load t+1}
\phi\Hmat_{\mathcal{L},\mathcal{V}}\bt_{t},~\phi\sim\mathcal{N}(\boldsymbol{1},\sigma^2_s\Imat),
\ee
and then computing $\bt_{t+1}$ by the rundcpf($\cdot$) Matpower command in \cite{matpower}.
In addition, throughout the simulations the noise is modeled as in \eqref{eq; FDIA model} with $\sigma_e^2=0.01$.

\subsubsection{Methods}
The proposed methods include the structural-constrained GIC method from Section \ref{sec; structural model}, 
and the low-complexity methods in Section \ref{sec; large scale}: the structural-constrained OMP method and the GM-GIC method. 
The penalty function for the GIC and GM-GIC methods defined in \eqref{eq; GIC} is set by
\be \label{eq; tuning}
\tau(\lvert \Lambda_i \rvert, \mathcal{L})= \zeta\lvert\Lambda_i\rvert + \gamma^{\text{GIC}}\delta[\lvert\Lambda_i\rvert],
\ee  
where $\zeta$ and $\gamma^{\text{GIC}}$ are user-defined regularization parameters and $\delta[\cdot]$ is the Kronecker delta function.
The term $\zeta\lvert\Lambda_i\rvert$ is set to encourage sparse solutions for the identification problem,
while the term $\gamma^{\text{GIC}}\delta(\lvert\Lambda_i\rvert)$ is set to maintain a desired false alarm rate for the detection problem. 
In particular, the probability of false alarm is set to $P_{FA}=0.05$ unless stated otherwise.
Furthermore, by selecting the AIC in \eqref{eq; MDL AIC}, the GIC tuning parameter in \eqref{eq; tuning} is set to $\zeta=2$.
The maximal sparsity rate is set to $K_{\cvec}=6$.

The performance of the proposed methods is compared with the following methods that were all modified according to the difference-based model in \eqref{eq; 1}:
\begin{enumerate}[label=\textbf{M.\arabic*}, leftmargin=1cm]
	\item \label{M; l2}
	The sparse optimization technique in \cite{gao2016identification}, which considers an $\ell_2$ relaxation
	for the attack sparsity restriction (denoted as $\ell_2$). 
	\item \label{M; l1} The sparse optimization technique in \cite{liu2014detecting}, which considers an $\ell_1$ relaxation
	for the attack sparsity restriction (denoted as $\ell_1$).
	\item \label{M; GFT} The GFT-based detector in \cite{drayer2018detection} (denoted as GFT).
	\item \label{M; ENG} The energy detector, inspired by the detector in \cite{jiang2017defense}, which is obtained by comparing  $T^{\text{ENG}}=\lvert \lvert \Delta\zvec\rvert \rvert^2$ 
	with a chosen threshold (denoted as ENG).
	\item \label{M; BDD} The BDD detector in \eqref{eq; BDD} (denoted as BDD).
\end{enumerate}
All numerical results in this section were obtained using at least 500 Monte Carlo simulations. 
The detection thresholds were computed from simulated historic data obtained by 500 off-line simulations of \eqref{eq; 1}
under the null hypothesis. 
The simulations where conducted using Matlab on an Intel(R) Xeon(R) CPU E5-2660 v4@ 2.00 GHz. 

\subsection{Simulations}
\label{sec; sim small scale}
\subsubsection{Detection} 
In Fig. \ref{fig; ROC}. the receiver operating characteristics (ROC)
curves of the proposed methods: GIC, GM-GIC, and OMP, 
are presented and compared with those of the methods \ref{M; l2}-\ref{M; BDD}. 
The simulations were conducted on the IEEE-30 bus test case with an attack on four state variables, $K_a=4$, 
where the attack norm is normalized to $\lvert\lvert\avec\rvert\rvert=0.05 K_{\avec}$. 
The loads scaling variation in \eqref{eq; load t+1} is set to $\sigma_{s}^2=0.05$.
This figure shows that the detection performance of the proposed low-complexity methods, GM-GIC, and OMP, converges to the performance of the optimal GIC method.
Furthermore, the proposed GIC, GM-GIC, and OMP methods provide a higher detection probability for any chosen false alarm rate, when compared to the previous methods. 
It should be noted that the sparse $\ell_2$ and $\ell_1$ methods were developed based on models with multiple time measurements
and where the sparsity pattern of the attack is defined over both the time and the bus domains. 
This is in contrast with the considered model, in which only two time samples are provided and where the attack is only sparse in the bus domain.
Finally, as expected from Subsection \ref{sec; sub; unobservable FDI attacks}, the BDD method cannot detect unobservable FDI attacks and is no better than flipping a coin.

\begin{figure}[t]
	\centering
	\includegraphics[width=7cm]{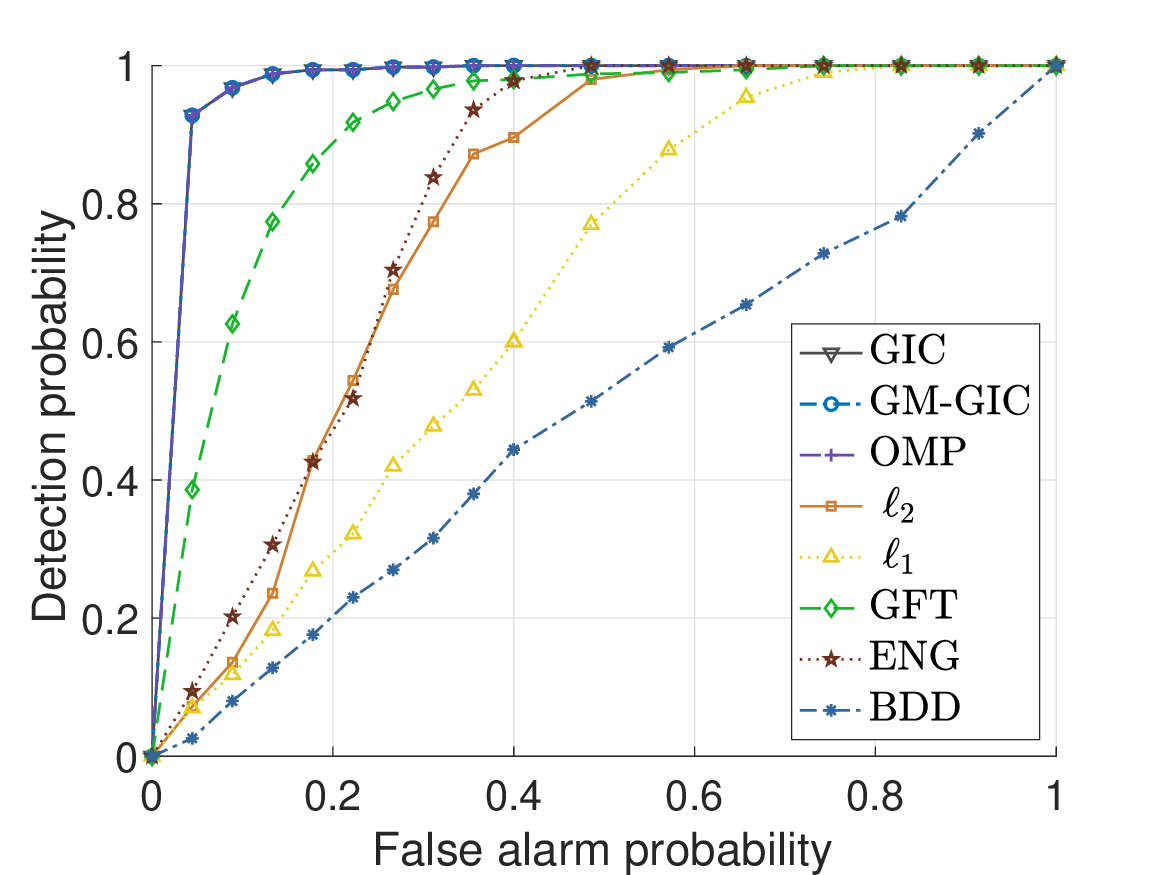}
	\caption{ Unobservable FDI attack detection: The ROC curve of the different methods 
		for
		$K_{\avec}=4$, ${\frac{\lvert\lvert\avec\rvert\rvert}{K_{\avec}}=0.05}$, and $\sigma_s^2=0.05$.}
	\label{fig; ROC}
\end{figure}

\subsubsection{Identification}
The identification performance 
is measured by the capability to classify each state variable as manipulated or not. 
Therefore, we evaluate the identification performance by the F-score classification metric \cite{sokolova}:
\be \label{eq; F score}
FS(\hat{\Lambda},\Lambda)=\frac{2t_p}{2t_p+f_n+f_p},
\ee
where $\Lambda$ is the true support of the state attack vector, $\cvec$, 
and $\hat{\Lambda}$ is the estimated support by a given method.
The terms $t_p$, $f_p$, and $f_n$, are the true-positive, false-positive, and false-negative probabilities, respectively.
The F-score metric takes values between $0$ and $1$, where  $1$
means perfect identification.

In Fig. \ref{fig; identification}, the identification performance, measured by the F-score metric, of the GIC, GM-GIC, OMP, $\ell_2$, and $\ell_1$ methods 
is examined w.r.t the attack characteristics.  
Methods \ref{M; GFT}-\ref{M; BDD} are not included in the comparison since they provide solely detection and do not identify the attacked buses' locations.
In Fig. \ref{fig; identification}.a the F-score is presented versus the normalized attack norm, $\frac{\lvert\lvert\avec\rvert\rvert}{K_{\avec}}$, for $K_{\avec}=4$ and $\sigma_s^2=0.1$.
In addition, in Fig. \ref{fig; identification}.b the F-score is presented versus the number of attacked elements, $K_{\avec}$, for $\lvert\lvert \avec\rvert\rvert=1.2$ and $\sigma_s^2=0.05$.
It can be observed from both figures that the proposed methods have a significantly higher F-score than those of the previous sparsity-based methods, $\ell_2$ and $\ell_1$. 
In addition, it can be seen from Fig. \ref{fig; identification}.a that the performance of all the methods
improves with the increase of the attack measured by $\frac{\lvert \lvert a\rvert \rvert }{K_{\avec}}$. 
Similarly, Fig. \ref{fig; identification}.b shows that the performance of all the methods degrades as $K_{\avec}$ increases.
Nonetheless, the F-score of the proposed methods is above $0.8$
even when a considerable portion ($>0.2$) of the system is attacked. 
Finally, it can be seen that both the GM-GIC and OMP methods
show relatively close results to those of the GIC method, 
where the GM-GIC method provides higher identification rates than the OMP method. 

\begin{figure}[t]
	\centering
	\subfigure[]{
		\includegraphics[width=7cm]{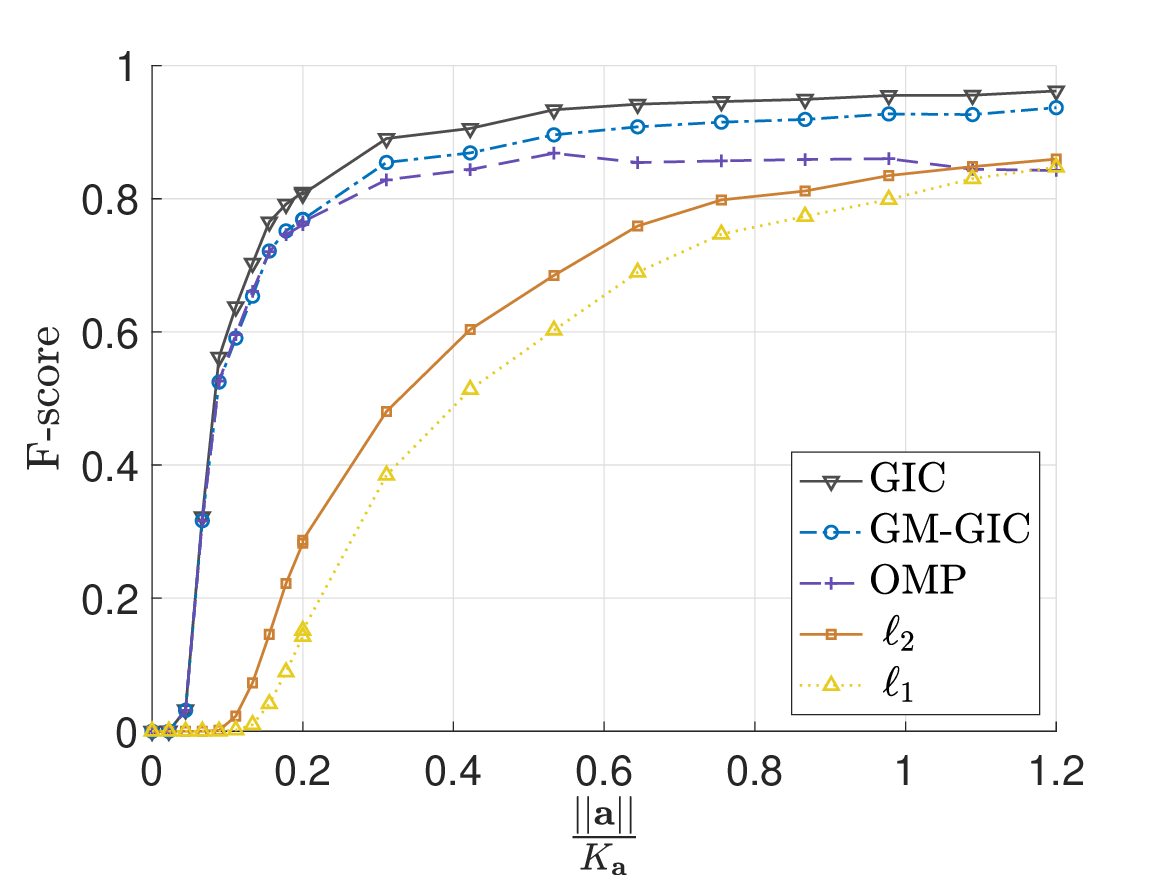}
		\label{fig; T_30_FS}
	}
	\subfigure[]{	
		\includegraphics[width=7cm]{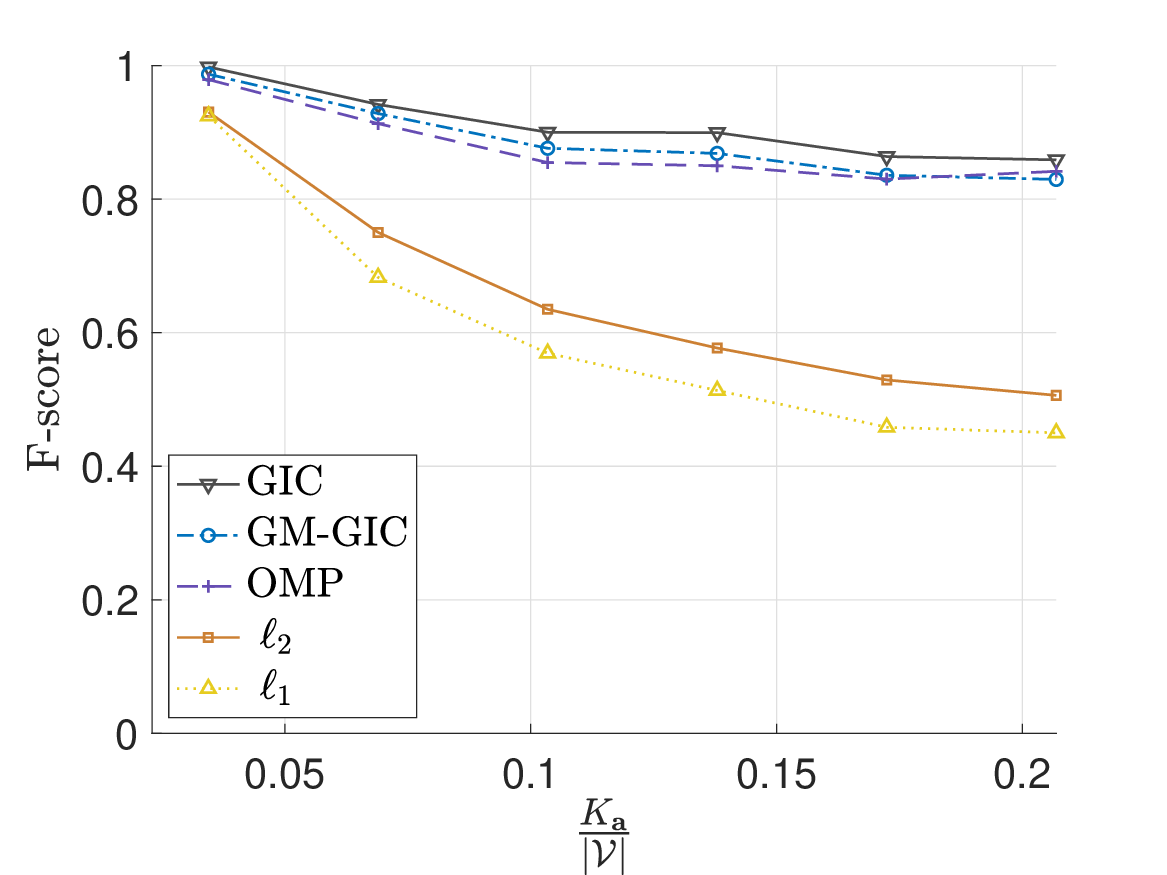}
		\label{fig; FS over Ka}
	}
	\caption{ Unobservable FDI attack identification: (a) The F-score of the different methods
		versus the normalized attack, $\frac{\lvert \lvert a\rvert \rvert}{K_{\avec}}$, where $K_{\avec}=4$ and $\sigma_s^2=0.1$; (b) 
		versus the ratio of attacked elements, $\frac{K_{\avec}}{\mathcal{V}}$, for $\lvert\lvert\avec\rvert\rvert=1.2$ and  $\sigma_s^2=0.05$.
	}
	\label{fig; identification}
\end{figure}

\subsubsection{Identification under mismatch}
The use of the GIC approach for identifying and detecting unobservable FDI attacks
is based on the assumption that the nuisance parameter  $\Hmat\Delta\bt$ is bounded and small, as described in \eqref{eq; quasi-static}. 
In Fig. \ref{fig; sigma}, we illustrate the robustness of the performance to this assumption by evaluating the influence of the norm of the nuisance parameter, $\lvert \lvert\Hmat\Delta\bt\rvert\rvert$, on the identification performance of the different methods. 
To this end, we consider the worst case, in which the equality in \eqref{eq; quasi-static} holds, i.e. $\lvert\lvert\Hmat_{\mathcal{L},\mathcal{V}}\Delta\bt\rvert\rvert^2=\eta$.
In particular, the F-score of the different methods, as well as of the clairvoyant GIC from \eqref{eq; GIC oracle}, is compared for different rates of 
$\eta$ for $K_{\avec}=4$ and $\lvert\lvert\avec\rvert\rvert=0.05 K_{\avec}$.

Fig. 4 shows that the F-score of the different methods, excluding the clairvoyant GIC from \eqref{eq; GIC oracle}, decreases as $\eta$ increases, due to the model mismatch. 
The F-score of the clairvoyant GIC is independent of $\eta$ since it uses the true value of $\Delta\bt$ and is not based on the approximated model, defined by $\eta$. 
It can be seen that this degradation in the identification performance 
is significantly greater for the sparsity-based methods, $\ell_2$ and $\ell_1$, than for the proposed methods.
Therefore, the proposed methods are more robust to a mismatched model
in comparison to the existing methods.
Finally, for small values of $\eta$ the proposed GIC method converges to the clairvoyant GIC. 

\begin{figure}[t]
	\centering
	\includegraphics[width=7cm]{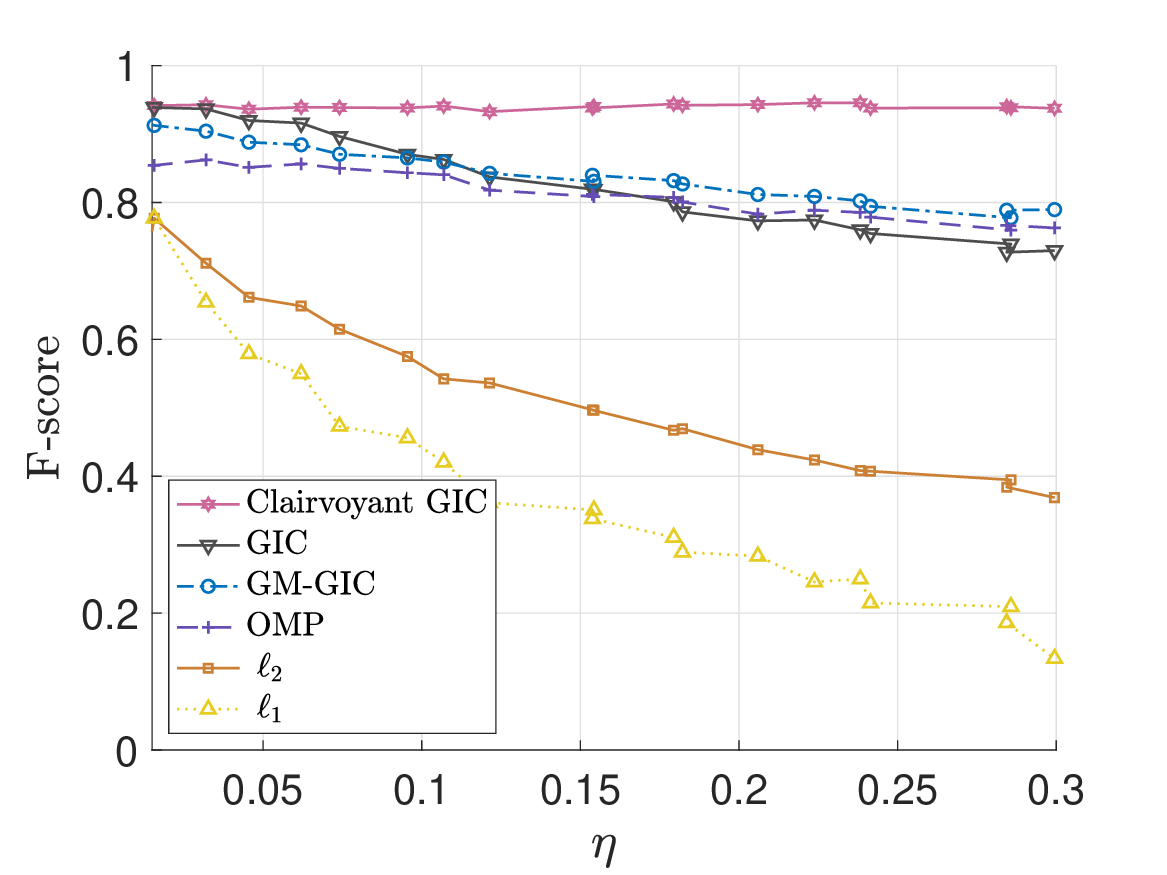}
	\caption{Unobservable FDI attack identification: The F-score of the different methods with $K_{\avec}=4$ and $\lvert\lvert\avec\rvert\rvert=0.05 K_{\avec}$.
	}
	\label{fig; sigma}
\end{figure}

\subsubsection{Run-time}
\label{sec; sim time}
In Fig. \ref{fig; run-time}, the averaged run-time of the identification methods is presented versus the ratio of node elements attacked, $\frac{K_{\avec}}{\lvert\mathcal{V}\rvert}$, 
for $\sigma_s^2=0.05$ and ${\lvert\lvert\avec\rvert\rvert=1.2}$.
It can be seen that the GIC method requires a run-time
which is significantly higher than those of the other methods. 
This makes the GIC method impractical for large networks. 
The  CS methods \cite{gao2016identification,liu2014detecting} have a higher averaged run-time than the proposed low-complexity methods, GM-GIC and OMP. 
As expected from the discussion in Subsection \ref{sec; GM-GIC small}, the GM-GIC method requires a significantly lower run-time 
than the GIC method, but  a higher run-time than the OMP method. 

\begin{figure}[hbt]
	\centering
	\includegraphics[width=7cm]{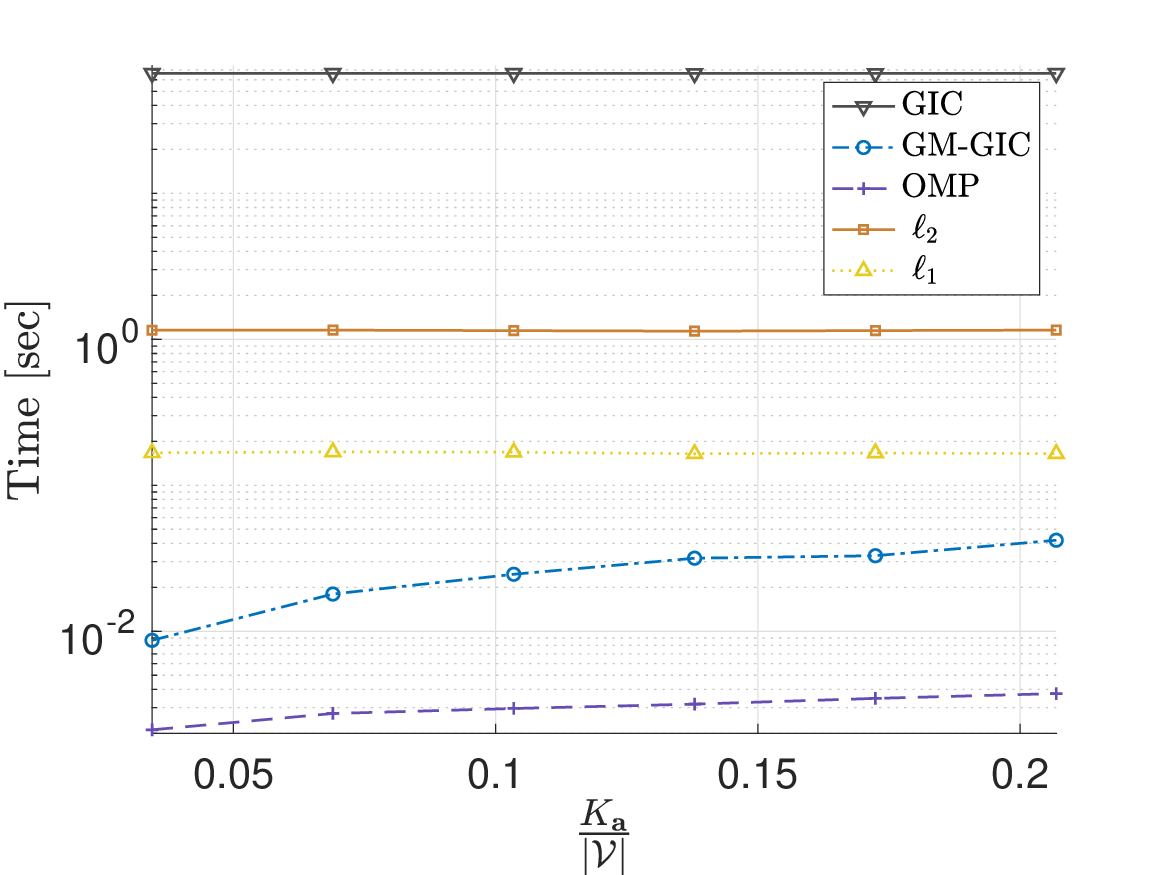}
	\caption{ Unobservable FDI attack identification: Run-time of the different methods versus the ratio of attacked elements, $\frac{K_{\avec}}{\mathcal{V}}$,
		for $\lvert\lvert\avec\rvert\rvert=1.2$ and $\sigma_s^2=0.05$.}
	\label{fig; run-time}
\end{figure}

\section{Conclusions}
\label{sec; Conclusions}
In this paper, we introduce novel methods for the identification of unobservable FDI attacks in power systems.
Unlike classical BDD methods that depend on measurement residuals and fail to detect unobservable FDI attacks, 
the proposed methods successfully detect and identify these attacks by utilizing power system intrinsic structural information. 
Specifically, the proposed methods are based on defining structural constraints on both the attack and the typical load changes,
and then formulating the identification problem as a model selection problem.
We develop the GIC selection rule for the identification task, as well as two low-complexity methods: 
1) the structural-constrained OMP method, which is a modification to the standard OMP method that accounts for the proposed structural and sparse constraints; and
2) the novel GM-GIC method that exploits the graph sparsity, graph Markovity, and the local behavior of graph connectivity in power systems. 
As by-products of the identification process, the proposed methods also enable the detection of unobservable FDI attacks 
and PSSE in the presence of these attacks. 
The low-complexity methods also enable fast implementation that can be integrated into an adaptive scheme that performs detection continuously to see if an attack is present at any given time.
In addition, we show the relations between the assumed problem and the problems of sparse recovery of graph signals and MSD. 
In particular, the proposed GM-GIC method can be applied for denoising an output of a general low-order GSP filter with a sparse input.   
This problem arises in a variety of GSP applications, such as epidemiology and anomaly detection in sensor networks.
In addition, performance analysis of the proposed approach is investigated by introducing a clairvoyant and an oracle detector of the considered FDI detection problem.
Our numerical simulations show that the proposed methods outperform existing methods for detection and identification, 
and are robust to the model assumptions. 
Moreover, the GM-GIC method, that integrates the structural information regarding the underling graph behind the data, 
obtains better performance than that of the OMP method and the existing sparsity-based methods.
The OMP and GM-GIC methods require  shorter run-time than the GIC method. 
Finally, the performance of the proposed approaches is compared with that of derived clairvoyant and oracle detectors. 

\begin{appendices}	
	\section{GIC statistic}
	\label{App; proof log likelihood}
	
	Based on \eqref{eq; model selection problem} and the noise statistics, $\Delta \evec_{\mathcal{L}}\sim\mathcal{N}(0,\sigma_e^2 \Imat)$,
	under each hypothesis, $i$, 
	the log-likelihood function of the measurements, for a given $\Delta\bt$, is 
	\be  \label{eq; log_like}
	\begin{aligned}
		L(\Delta\zvec_{\mathcal{L}}; \cvec_{\Lambda_i},\Delta\bt) &=
		-\frac{\lvert\mathcal{L}\rvert}{2}\ln2\pi\sigma_e^2\hspace{3cm}\\
		&~~~~-\frac{1}{2\sigma_{\evec}^2}\lvert\lvert\Delta\zvec_{\mathcal{L}}-\Hmat_{\mathcal{L},\Lambda_i}\cvec_{\Lambda_i}-\Hmat_{\mathcal{L},\mathcal{V}} \Delta\bt \rvert\rvert^2.
	\end{aligned}
	\ee
	Thus, the ML estimator of the state attack vector, $\cvec_{\Lambda_i}$,  is obtained by maximizing \eqref{eq; log_like} w.r.t. $\cvec_{\Lambda_i}$.
	Since \eqref{eq; log_like} is a concave function w.r.t. $\cvec_{\Lambda_i}$, 
	the ML estimator can be computed by equating the derivative of \eqref{eq; log_like}
	w.r.t. $\cvec_{\Lambda_i}$ to zero, as follows:
	\be \label{eq; likelihood differentiation}
	\begin{aligned}
		\frac{\partial L(\Delta\zvec_{\mathcal{L}}; \cvec_{\Lambda_i}) }{\partial \cvec_{\Lambda_i}} \bigg\rvert_{\cvec_{\Lambda_i}= \hat{\cvec}^{\text{ML}|i}_{\Lambda_i}} \hspace{4.5cm} \\
		= \frac{1}{2\sigma_{\evec}^2}\Hmat_{\mathcal{L},\Lambda_i}^T(\Delta\zvec_{\mathcal{L}}-\Hmat_{\mathcal{L},\Lambda_i}\hat{\cvec}^{\text{ML}|i}_{\Lambda_i} -\Hmat_{\mathcal{L},\mathcal{V}} \Delta\bt)=\boldsymbol{0},
	\end{aligned}
	\ee
	which implies that the ML estimator is
	\be \label{eq; estimated state attack vector} 
	\begin{aligned} \hat{\cvec}^{\text{ML}|i}_{\Lambda_i} 
		&= ( \Hmat_{\mathcal{L},\Lambda_i} ^T \Hmat_{\mathcal{L},\Lambda_i} )^{-1} \Hmat_{\mathcal{L},\Lambda_i}^T  (\Delta\zvec_{\mathcal{L}}- \Hmat_{\mathcal{L},\mathcal{V}}  \Delta\bt ). \\
	\end{aligned}
	\ee
	By using the projection matrix from \eqref{eq; projection matrix}, $\Pmat_{\Lambda_i}$, and its orthogonal projection, $\Pmat_{\Lambda_i}^{\bot}$, we can decompose $\Hmat_{\mathcal{L},\mathcal{V}} \Delta\bt$ as follows:
	\be \label{eq; decomposition of loads}
	\Hmat_{\mathcal{L},\mathcal{V}} \Delta\bt=\Pmat_{\Lambda_i}\Hmat_{\mathcal{L},\mathcal{V}} \Delta\bt+\Pmat_{\Lambda_i}^{\bot}\Hmat_{\mathcal{L},\mathcal{V}} \Delta\bt.
	\ee
	By substituting the constraint from \eqref{eq; model selection problem}, $\Pmat_{\Lambda_i}\Hmat_{\mathcal{L},\mathcal{V}}\Delta\bt  = \boldsymbol{0}$, 
	in \eqref{eq; decomposition of loads}, one obtains
	\be \label{eq; decomposition of loads 2}
	\Hmat_{\mathcal{L},\mathcal{V}} \Delta\bt=\Pmat_{\Lambda_i}^{\bot}\Hmat_{\mathcal{L},\mathcal{V}} \Delta\bt.
	\ee
	Substitution of \eqref{eq; decomposition of loads 2} in \eqref{eq; estimated state attack vector} results in 
	\be \label{eq; estimated state attack vector 2}
	\begin{aligned}
		\hat{\cvec}^{\text{ML}|i}_{\Lambda_i} 
		&= ( \Hmat_{\mathcal{L},\Lambda_i} ^T \Hmat_{\mathcal{L},\Lambda_i} )^{-1}   \Hmat_{\mathcal{L},\Lambda_i}^T(\Delta\zvec_{\mathcal{L}}-\Pmat_{\Lambda_i}^{\bot}\Hmat_{\mathcal{L},\mathcal{V}} \Delta\bt) \\
		&= ( \Hmat_{\mathcal{L},\Lambda_i} ^T \Hmat_{\mathcal{L},\Lambda_i} )^{-1}   \Hmat_{\mathcal{L},\Lambda_i}^T\Delta\zvec_{\mathcal{L}},
	\end{aligned}
	\ee
	where the last equality is obtained from the properties of projection matrices.
	The generalized log likelihood is obtained by substituting \eqref{eq; decomposition of loads 2} and \eqref{eq; estimated state attack vector 2} in \eqref{eq; log_like}: 
	\be \label{eq; log_like_wc}
	\begin{aligned}
		L(\Delta\zvec_{\mathcal{L}};\hat{\cvec}^{\text{ML}|i}_{\Lambda_i}, \Delta\bt)   
		&= -\frac{\lvert\mathcal{L}\rvert}{2}\ln 2\pi\sigma_e^2 \\
		&~~~~-\frac{1}{2\sigma_{\evec}^2}\lvert\lvert \Pmat_{\Lambda_i}^{\bot}(\Delta\zvec_{\mathcal{L}} - \Hmat_{\mathcal{L},\mathcal{V}}\Delta\bt) 
		\rvert\rvert^2.
	\end{aligned}
	\ee
	Additionally, from the properties of projection matrices, it can be verified that (see, e.g. p. 46 in \cite{yanai2011projection}) 
	\be \label{eq; projection matrix norm property 2}
	\begin{aligned}
		\lvert\lvert \Pmat_{\Lambda_i}^{\bot}(\Delta\zvec_{\mathcal{L}}- \Hmat_{\mathcal{L},\mathcal{V}}\Delta\bt)  \rvert\rvert^2 = \\
		& \hspace{-2.5cm} 	\lvert\lvert (\Delta\zvec_{\mathcal{L}}- \Hmat_{\mathcal{L},\mathcal{V}}\Delta\bt)  \rvert\rvert^2 
		-\lvert\lvert \Pmat_{\Lambda_i}(\Delta\zvec_{\mathcal{L}}- \Hmat_{\mathcal{L},\mathcal{V}}\Delta\bt)  \rvert\rvert^2. 
	\end{aligned}
	\ee
	By substituting the constraint from \eqref{eq; model selection problem}, $	\Pmat_{\Lambda_i}\Hmat_{\mathcal{L},\mathcal{V}}\Delta\bt  = \boldsymbol{0}$, 
	in \eqref{eq; projection matrix norm property 2} we obtain
	\be \label{eq; projection matrix norm property}
	\begin{aligned}
		\lvert\lvert \Pmat_{\Lambda_i}^{\bot}(\Delta\zvec_{\mathcal{L}}- \Hmat_{\mathcal{L},\mathcal{V}\Delta\bt)}  \rvert\rvert^2= 	 \\
		&\hspace{-2.5cm} \lvert\lvert (\Delta\zvec_{\mathcal{L}}- \Hmat_{\mathcal{L},\mathcal{V}}\Delta\bt)  \rvert\rvert^2
		-\lvert\lvert \Pmat_{\Lambda_i}\Delta\zvec_{\mathcal{L}}  \rvert\rvert^2,
	\end{aligned}
	\ee
	$\forall \Lambda_i\in \mathcal{G}_{K_{\cvec}}.$
	By substituting \eqref{eq; projection matrix norm property} in \eqref{eq; log_like_wc} we obtain 
	\be \label{eq; log_like_wc 2} 
	\begin{aligned}
		L(\Delta\zvec_{\mathcal{L}};\hat{\cvec}^{\text{ML}|i}_{\Lambda_i}, \Delta\bt) &= -\frac{\lvert\mathcal{L}\rvert}{2}\ln 2\pi\sigma_e^2 \\
		&\hspace{-2.5cm}  -\frac{1}{2\sigma_{\evec}^2}\lvert\lvert (\Delta\zvec_{\mathcal{L}}- \Hmat_{\mathcal{L},\mathcal{V}}\Delta\bt)  \rvert\rvert^2+\frac{1}{2\sigma_{\evec}^2}\lvert\lvert \Pmat_{\Lambda_i}\Delta\zvec_{\mathcal{L}}  \rvert\rvert^2,
	\end{aligned}
	\ee
	$\forall \Lambda_i\in \mathcal{G}_{K_{\cvec}}.$
	By substituting \eqref{eq; log_like_wc 2} into \eqref{eq; GIC function},
	we obtain that the GIC statistic satisfies
	\be \label{eq; GIC prior}
	\begin{aligned}
		\text{GIC}(\Lambda_i,\tau(\lvert\Lambda_i\rvert, \lvert\mathcal{L}\lvert)  ) =\\
		&\hspace{-2cm} -\lvert\mathcal{L}\rvert\ln 2\pi\sigma_e^2 -\frac{1}{\sigma_{\evec}^2}\lvert\lvert (\Delta\zvec_{\mathcal{L}}- \Hmat_{\mathcal{L},\mathcal{V}}\Delta\bt)\rvert\rvert^2  \\
		&\hspace{-2cm}+ \frac{1}{\sigma_{\evec}^2}\lvert\lvert \Pmat_{\Lambda_i}\Delta\zvec_{\mathcal{L}}  \rvert\rvert^2 -\tau(\lvert\Lambda_i\rvert, \lvert\mathcal{L}\lvert),
	\end{aligned}
	\ee
	where $~{i=0,1,\ldots,\lvert\mathcal{G}_{K_{\cvec}}\rvert}$.
	Since the term $$-\lvert\mathcal{L}\rvert\ln 2\pi\sigma_e^2 -\frac{1}{\sigma_{\evec}^2}\lvert\lvert (\Delta\zvec_{\mathcal{L}}- \Hmat_{\mathcal{L},\mathcal{V}}\Delta\bt)\rvert\rvert^2$$
	is independent of the hypothesis, i.e. is not a function of $\Lambda_i$, maximizing the GIC rule from \eqref{eq; GIC prior} w.r.t. $i=0,1,\ldots,\lvert\mathcal{G}_{K_{\cvec}}\rvert$, is equivalent to
	maximizing the r.h.s. of \eqref{eq; GIC} w.r.t. the same candidate sets, $i=0,1,\ldots,\lvert\mathcal{G}_{K_{\cvec}}\rvert$. 
	
	\section{Oracle GLRT} \label{App; discussion}
	While the GLRT in \eqref{eq; GLRT binary 2} was developed for the  binary modified hypothesis testing
	problem in \eqref{eq; model selection problem}, 
	, it should be analyzed based on the binary version of the original hypothesis testing  in
	\eqref{eq; model selection problem the} with only $\Hnull$ and a single $\mathcal{H}_i$. 
	In this case, the GLRT detector is distributed as follows 
	(see, e.g. Appendix $7$B in \cite{kayfundamentals}):
	\be \label{eq; chi} 
	\begin{cases}
		\mathcal{H}_0:&T^{GLRT|i} \sim   \chi^2_{\lvert\Lambda_i\rvert} \big(    \frac{\lvert\lvert\Pmat_{\Lambda_i}\Hmat_{\mathcal{L},\mathcal{V}}\Delta\bt) \rvert\rvert^2}{\sigma_{\evec}^2} \big ) \\
		\mathcal{H}_i:&T^{GLRT|i} \sim   \chi^2_{\lvert\Lambda_i\rvert} \big(   \frac{\lvert\lvert\Pmat_{\Lambda_i}(\Hmat_{\mathcal{L},\Lambda_i}\cvec_{\Lambda_i}  +\Hmat_{\mathcal{L},\mathcal{V}}\Delta\bt) \rvert\rvert^2}{\sigma_{\evec}^2} \big ),
	\end{cases}
	\ee
	where $\chi^2_r(\lambda)$ denotes a non-central $\chi$-square distribution with $r$ degrees of freedom and a non-centrality parameter $\lambda$. 
	
	From \eqref{eq; chi} it can be verified that the probability of false alarm and the probability of detection of the GLRT from \eqref{eq; GLRT binary 2} are given by (see Subsection 2.2.3 in \cite{kayfundamentals}):
	\be \label{eq; false}
	P_{fa}=\mathcal{Q}_{\frac{\lvert\Lambda_i\rvert}{2}}  \left( \frac{\lvert\lvert\Pmat_{\Lambda_i}\Hmat_{\mathcal{L},\mathcal{V}}\Delta\bt \rvert\rvert}{\sigma_e},\sqrt{\gamma^{GLRT} }\right)
	\ee
	and 
	\be \label{eq; detect}
	P_{d}=\mathcal{Q}_{\frac{\lvert\Lambda_i\rvert}{2}}  \left(\frac{\lvert\lvert\Pmat_{\Lambda_i}(\Hmat_{\mathcal{L},\Lambda_i}\cvec_{\Lambda_i}  +\Hmat_{\mathcal{L},\mathcal{V}}\Delta\bt) \rvert\rvert}{\sigma_e},\sqrt{\gamma^{GLRT} } \right),
	\ee
	respectively, where $\mathcal{Q}_{v}(a,b)$ is the generalized Marcum Q-function of order $v$ \cite{sun2010monotonicity}.
	From \cite{sun2010monotonicity}, $\mathcal{Q}_{v}(a,b)$ strictly increases as $a$ increases. 
	That is, 
	\be \label{eq; mono}
	\mathcal{Q}_{v}(\sqrt{a_1},\sqrt{b})<\mathcal{Q}_{v}(\sqrt{a_1+a_2},\sqrt{b})
	\ee
	for all $a_1\ge 0$ and $a_2,v,b>0$.
	Thus, the probability of false alarm in \eqref{eq; false} strictly increases as $ \lvert\lvert\Pmat_{\Lambda_i}\Hmat_{\mathcal{L},\mathcal{V}}\Delta\bt \rvert\rvert$ increases.
	Consequently, by considering the worst case in \eqref{eq; projection loads} of $\lvert\lvert\Pmat_{\Lambda_i}\Hmat_{\mathcal{L},\mathcal{V}}\Delta\bt) \rvert\rvert^2=\epsilon_i\eta$,
	we obtain that the probability of false alarm in \eqref{eq; false} strictly increases as $ \epsilon_i\eta$ increases.
	
	The following proposition defines upper and lower bounds on the approximation error $\epsilon_i$. 
	\begin{prop} \label{prop; approximation}
		The approximation error from \eqref{eq; approximation error} can be bounded by 
		\be\label{eq; approx bound}
		0\le \epsilon_i \le 1. 
		\ee
	\end{prop}
	\begin{proof}
		By using the properties of projection matrices it can be verified (see Theorem 2.22 in \cite{yanai2011projection}) that
		\be \label{eq; projection bound}
		0\le \lvert\lvert\Pmat_{\Lambda_i}\Hmat_{\mathcal{L},\mathcal{V}}\Delta\bt\rvert\rvert^2 \le \lvert\lvert \Hmat_{\mathcal{L},\mathcal{V}}\Delta \bt \rvert\rvert^2.
		\ee
		In addition, Assumption \ref{A; loads} implies that the typical load changes are nonzero, i.e. $\Hmat_{\mathcal{L},\mathcal{V}}\Delta \bt \ne 0$. 
		By dividing the strictly positive term $\lvert\lvert \Hmat_{\mathcal{L},\mathcal{V}}\Delta \bt \rvert\rvert^2$ from the inequality in \eqref{eq; projection loads}  we obtain
		\be \label{eq; projection bound 2}
		0\le \frac{\lvert\lvert\Pmat_{\Lambda_i}\Hmat_{\mathcal{L},\mathcal{V}}\Delta\bt\rvert\rvert^2}{\lvert\lvert \Hmat_{\mathcal{L},\mathcal{V}}\Delta \bt \rvert\rvert^2} \le 1.
		\ee
		Thus, by substituting \eqref{eq; approximation error} in \eqref{eq; projection bound 2} we get \eqref{eq; approx bound}. 
	\end{proof}
	
	By applying the properties in \eqref{eq; mono} and \eqref{eq; approx bound}, from Proposition \ref{prop; approximation}, on \eqref{eq; false}, we obtain \eqref{eq; fa bound}. 
	
	\section{Single-node attack} \label{App; single attack}
	In Proposition 1 in \cite{sedghi2015statistical} it was shown that under the assumption that the injected powers in the nodes are  
	Gaussian distributed and mutually independent, 
	which holds under the model in \eqref{eq; 1},
	the columns of the topology matrix satisfy
	\be \label{eq; App 0}
	\Hmat_{\mathcal{V},m}^T\Hmat_{\mathcal{V},k}=0,~~\forall k,m\in\mathcal{V},~d(k,m)>2.
	\ee
	Similarly, the columns of the subtopology matrix associated with the load buses satisfy
	\be \label{eq; App 1}
	\Hmat_{\mathcal{L},m}^T\Hmat_{\mathcal{L},k}=0,~~\forall k,m\in\mathcal{V},~d(k,m)>2.
	\ee
	By substituting \eqref{eq; projection matrix} with $\Lambda_i=\{m\}$ in the l.h.s. of \eqref{eq; 2nd 0} we obtain
	\be \label{eq; App 2}
	\begin{aligned}
		\Pmat_{\{m\}}\Hmat_{\mathcal{L},k}c_k&=\Hmat_{\mathcal{L},m}(\Hmat_{\mathcal{L},m}^T\Hmat_{\mathcal{L},m})^{-1}\Hmat_{\mathcal{L},m}^T\Hmat_{\mathcal{L},k}c_k.
	\end{aligned}
	\ee
	Thus, by substituting \eqref{eq; App 1} in \eqref{eq; App 2} we obtain \eqref{eq; 2nd 0}.
	In addition, by using the properties of projection matrices it can be verified (see Theorem 2.22 in \cite{yanai2011projection}) that
	\be \label{eq; App2 3}
	\lvert\lvert \Pmat_{\{m\}}\Hmat_{\mathcal{L},k}\cvec_k \rvert\rvert\le \lvert\lvert\Hmat_{\mathcal{L},k}\cvec_k\rvert\rvert,~~\forall k,m\in\mathcal{V},
	\ee
	and
	\be \label{eq; App2 4}
	\Pmat_{\{k\}}\Hmat_{\mathcal{L},k}\cvec_k= \Hmat_{\mathcal{L},k}\cvec_k.
	\ee
	Hence, by substituting \eqref{eq; App2 4} in \eqref{eq; App2 3} we obtain \eqref{eq; all orders}.
\end{appendices}

\bibliographystyle{IEEEtran}
\bibliography{FDIdetection}

\end{document}